\documentclass[11pt,letterpaper]{article}

\usepackage[round,comma]{natbib}
\usepackage{times,mysetup,algorithm,my-acm}
\usepackage[noend]{algorithmic}

\FULLPAGE

\providecommand{\Appendix}{}
\renewcommand{\Appendix}[2][?]{%
	\refstepcounter{section}%
	\vspace{\parskip}%
	{\flushright\large\bfseries\appendixname\ \thesection: #1}%
	\vspace{\baselineskip}%
}

\renewcommand{\appendix}{%
	\renewcommand{\section}{\secdef\Appendix\Appendix}%
	\renewcommand{\thesection}{\Alph{section}}%
	\setcounter{section}{0}%
}

\newcommand{\D}{\mathcal{D}}  % for similarity distance
\newcommand{\mP}{\mathcal{P}} % for set of feasible (x,y) pairs
\newcommand{\mA}{\mathcal{A}}
\newcommand{\mB}{\mathcal{B}}
\newcommand{\mI}{\mathcal{I}}
\newcommand{\reward}{\mathtt{rew}}
\newcommand{\LipConst}{K_{\text{Lip}}}
\newcommand{\mPzoom}{\mP_{\mu,r}}
\newcommand{\mWzoom}{\mathcal{W}_{\mu,r}}
\newcommand{\Nadj}{N^{\text{adj}}}

\newcommand{\indicator}[1]{\ensuremath{\mathbf{1}_{\{ #1 \}}}} % the indicator function
\newcommand{\problem}{contextual MAB problem}

\newcommand{\driftProblem}{drifting MAB problem}

\newcommand{\TX}{{\text{X}}}
\newcommand{\TY}{{\text{Y}}}

\newcommand{\arr}{x_{(1..T)}} % context arrivals x_1..x_T
\newcommand{\Cdbl}{\ensuremath{c_{\textsc{dbl}}}} % doubling constant
\newcommand{\ZoomDim}{contextual zooming dimension}

\newcommand{\bandit}{{\tt Bandit}}
\newcommand{\contextBandit}{{\tt ContextualBandit}}

\newcommand{\naiveAlg}{uniform algorithm}
\newcommand{\ZoomAlg}{contextual zooming algorithm}

\newcommand{\EXP}[1][3]{\ensuremath{\textsc{exp{\small #1}}}}

\providecommand{\qedhere}{}

\newcommand{\TAB}{\hspace{5mm}}
\newcommand{\RELEVANT}{\mathtt{relevant}} % set of relevant balls

\title{\LARGE\bf Contextual Bandits with Similarity Information%
\thanks{This is the full version of a conference paper in \emph{COLT 2011}.
A preliminary version of this manuscript has been posted to {\tt arxiv.org} in February 2011. An earlier version on {\tt arxiv.org}, which does not include the results in Section~\ref{sec:apps}, dates back to July 2009.
The present revision addresses various presentation issues pointed out by journal referees.}}

\author{\Large Aleksandrs Slivkins
\thanks{Microsoft Research New York, New York, NY 10011, USA.
Email:~{\tt slivkins at microsoft.com}.}}

\date{First version: February 2009\\ This revision: May 2014}

\begin{document}
\maketitle

\begin{abstract}
In a multi-armed bandit (MAB) problem, an online algorithm makes a sequence of choices. In each round it chooses from a time-invariant set of alternatives and receives the payoff associated with this alternative. While the case of small strategy sets is by now well-understood, a lot of recent work has focused on MAB problems with exponentially or infinitely large  strategy sets, where one needs to assume extra structure in order to make the problem tractable. In particular, recent literature considered information on similarity between arms.

We consider similarity information in the setting of \emph{contextual bandits}, a natural extension of the basic MAB problem where before each round an algorithm is given the \emph{context} -- a hint about the payoffs in this round. Contextual bandits are directly motivated by placing advertisements on webpages, one of the crucial problems in sponsored search. A particularly simple way to represent similarity information in the contextual bandit setting is via a \emph{similarity distance} between the context-arm pairs which bounds from above the difference between the respective expected payoffs.

Prior work on contextual bandits with similarity uses ``uniform" partitions of the similarity space, so that each context-arm pair is approximated by the closest pair in the partition. Algorithms based on ``uniform" partitions disregard the structure of the payoffs and the context arrivals, which is potentially wasteful. We present algorithms that are based on {\em adaptive} partitions, and take advantage of "benign" payoffs and context arrivals without sacrificing the worst-case performance. The central idea is to maintain a finer partition in high-payoff regions of the similarity space and in popular regions of the context space. Our results apply to several other settings, e.g. MAB with constrained temporal change~\citep{DynamicMAB-colt08} and sleeping bandits~\citep{sleeping-colt08}.

\end{abstract}

{\bf ACM Categories and subject descriptors:}
\category{F.2.2}{Analysis of Algorithms and Problem Complexity}{Nonnumerical Algorithms and Problems}
\category{F.1.2}{Computation by Abstract Devices}{Modes of Computation}[Online computation]

\terms{Algorithms,Theory}

\keywords{online learning, multi-armed bandits, contextual bandits, regret minimization, metric spaces}
\newpage

%%%%%%%%%%%%
\section{Introduction}
\label{sec:intro}

In a multi-armed bandit problem (henceforth, ``multi-armed bandit" will be abbreviated as MAB), an algorithm is presented with a sequence of trials. In each round, the algorithm chooses one alternative from a set of alternatives (\emph{arms}) based on the past history, and receives the payoff associated with this alternative. The goal is to maximize the total payoff of the chosen arms. The MAB setting has been introduced in 1952 in~\cite{Robbins1952} and studied intensively since then in Operations Research, Economics and Computer Science. This setting is a clean model for the exploration-exploitation trade-off, a crucial issue in sequential decision-making under uncertainty.

One standard way to evaluate the performance of a bandit algorithm is {\it regret}, defined as the difference between the expected payoff of an optimal arm and that of the algorithm. By now the MAB problem with a small finite set of arms is quite well understood, e.g. see~\citet{Lai-Robbins-85,bandits-exp3,bandits-ucb1}. However, if the arms set is exponentially or infinitely large, the problem becomes intractable unless we make further assumptions about the problem instance. Essentially, a bandit algorithm needs to find a needle in a haystack; for each algorithm there are inputs on which it performs as badly as random guessing.

Bandit problems with large sets of arms have been an active area of investigation in the past decade (see Section~\ref{sec:related-work} for a discussion of related literature). A common theme in these works is to assume a certain \emph{structure} on payoff functions. Assumptions of this type are natural in many applications, and often lead to efficient learning algorithms \citep{Bobby-thesis}. In particular, a line of work started in~\citet{Agrawal-bandits-95} assumes that some information on similarity between arms is available.

In this paper we consider similarity information in the setting of \emph{contextual bandits}~\citep{Woodroofe79,Auer-focs00,Wang-sideMAB05,yahoo-bandits07,Langford-nips07}, a natural extension of the basic MAB problem where before each round an algorithm is given the \emph{context} -- a hint about the payoffs in this round. Contextual bandits are directly motivated by the problem of placing advertisements on webpages, one of the crucial problems in sponsored search. One can cast it as a bandit problem so that arms correspond to the possible ads, and payoffs correspond to the user clicks. Then the context consists of information about the page, and perhaps the user this page is served to. Furthermore, we assume that similarity information is available on both the context and the arms. Following the work in~\cite{Agrawal-bandits-95,Bobby-nips04,AuerOS/07,LipschitzMAB-stoc08} on the (non-contextual) bandits, a particularly simple way to represent similarity information in the contextual bandit setting is via a \emph{similarity distance} between the context-arm pairs, which gives an upper bound on the difference between the corresponding payoffs.

%%%%%%%%%%%%%%%%%%%%%%%%%%%%%%%%%%%%%
\xhdr{Our model: contextual bandits with similarity information.}
%\label{sec:model-contextual}

The contextual bandits framework is defined as follows. Let $X$ be the \emph{context set} and $Y$ be the \emph{arms set}, and let $\mP\subset X\times Y$ be the set of feasible context-arms pairs. In each round $t$, the following events happen in succession:
\begin{enumerate}
\item a context $x_t\in X$ is revealed to the algorithm,
\item the algorithm chooses an arm $y_t\in Y$ such that $(x_t,y_t) \in \mP$,
\item payoff (reward) $\pi_t\in [0,1]$ is revealed.
\end{enumerate}

\noindent The sequence of context arrivals $(x_t)_{t\in\N}$ is fixed before the first round, and does not depend on the subsequent choices of the algorithm. With \emph{stochastic payoffs}, for each pair $(x,y)\in \mP$ there is a distribution $\Pi(x,y)$ with expectation $\mu(x,y)$, so that $\pi_t$ is an independent sample from $\Pi(x_t,y_t)$. With \emph{adversarial payoffs}, this distribution can change from round to round. For simplicity, we present the subsequent definitions for the stochastic setting only, whereas the adversarial setting is fleshed out later in the paper (Section~\ref{sec:CovAlg}).

\OMIT{Here $\pi_t: X\times Y\to [0,1]$ is the \emph{payoff function} defined as an independent random sample from some fixed distribution $\Pi$ over functions $X\times Y\to [0,1]$.}

In general, the goal of a bandit algorithm is to maximize the total payoff $\sum_{t=1}^T \pi_t$, where $T$ is the \emph{time horizon}. In the contextual MAB setting, we benchmark the algorithm's performance in terms of the context-specific ``best arm". Specifically, the goal is to minimize the \emph{contextual regret}:
\begin{align*}
%\label{eq:regret-defn}
R(T) \triangleq \textstyle{\sum_{t=1}^T}\;
	 \mu(x_t, y_t) - \mu^*(x_t) ,
	\quad\text{where}\quad
\mu^*(x) \triangleq \textstyle{\sup_{y\in Y:\, (x,y)\in \mP}}\; \mu(x,y).
\end{align*}

The context-specific best arm is a more demanding benchmark than the best arm used in the ``standard" (context-free) definition of regret.

%Here $\mu^*_t(x)$ is the \emph{benchmark payoff} in round $t$.

The similarity information is given to an algorithm as a metric space $(\mP,\D)$ which we call the \emph{similarity space}, such that the following Lipschitz condition\footnote{\label{fn:LipConst}In other words, $\mu$ is a Lipschitz-continuous function on $(X,\mP)$, with Lipschitz constant $\LipConst=1$. Assuming $\LipConst=1$ is without loss of generality (as long as $\LipConst$ is known to the algorithm), since we can re-define $\D \leftarrow \LipConst\, D$.} holds:
\begin{align}\label{eq:LipschitzD}
|\mu(x,y) - \mu(x',y')|
	&\leq \D( (x,y),\; (x',y')).
\end{align}
Without loss of generality, $\D\leq 1$. The absence of similarity information is modeled as $\D = 1$.

An instructive special case is the \emph{product similarity space}
	$ (\mP,\D) = (X\times Y, \D)$,
where $(X,\D_\TX)$ is a metric space on contexts (\emph{context space}), and $(Y,\D_\TY)$ is a metric space on arms (\emph{arms space}), and
\begin{align}\label{eq:product-space}
\D( (x,y),\; (x',y'))
    = \min(1,\; \D_\TX(x,x') + \D_\TY(y,y')).
\end{align}

%We term this problem the {\bf\em \problem}.

%%%%%%%%%%%%%%%%%%%%%%%%%%%%%%%%%%%%%
\xhdr{Prior work: uniform partitions.}
\cite{Hazan-colt07} consider contextual MAB with similarity information on contexts. They suggest an algorithm that chooses a ``uniform" partition $S_\TX$ of the context space and approximates $x_t$ by the closest point in $S_\TX$, call it $x'_t$. Specifically, the algorithm creates an instance $\A(x)$ of some bandit algorithm $\A$ for each point $x\in S_\TX$, and invokes $\A(x'_t)$ in each round $t$. The granularity of the partition is adjusted to the time horizon, the context space, and the black-box regret guarantee for $\A$. Furthermore, \cite{Bobby-nips04} provides a bandit algorithm $\A$ for the adversarial MAB problem on a metric space that has a similar flavor: pick a ``uniform" partition $S_\TY$ of the arms space, and run a $k$-arm bandit algorithm such as \EXP~\cite{bandits-exp3} on the points in $S_\TY$. Again, the granularity of the partition is adjusted to the time horizon, the arms space, and the black-box regret guarantee for \EXP.

Applying these two ideas to our setting (with the product similarity space) gives a simple
algorithm which we call the \emph{\naiveAlg}. Its contextual regret, even for adversarial payoffs, is
 \begin{align}\label{eq:regret-naive} 	
 R(T) \leq O( T^{1-1/(2+d_\TX+d_\TY)}) (\log T),
\end{align}
where $d_\TX$ is the covering dimension of the context space and $d_\TY$ is that of the arms space.

\xhdr{Our contributions.}
Using ``uniform" partitions disregards the potentially benign structure of expected payoffs and context arrivals. The central topic in this paper is {\bf\em adaptive partitions} of the similarity space which are adjusted to frequently occurring contexts and high-paying arms, so that the algorithms can take advantage of the problem instances in which the expected payoffs or the context arrivals are ``benign" (``low-dimensional"), in a sense that we make precise later.

We present two main results, one for stochastic payoffs and one for adversarial payoffs. For stochastic payoffs, we provide an algorithm called \emph{contextual zooming} which ``zooms in" on the regions of the context space that correspond to frequently occurring contexts, and the regions of the arms space that correspond to high-paying arms. Unlike the algorithms in prior work, this algorithm considers the context space and the arms space \emph{jointly} -- it maintains a partition of the similarity space, rather than one partition for contexts and another for arms. We develop provable guarantees that capture the ``benign-ness" of the context arrivals and the expected payoffs. In the worst case, we match the guarantee~\refeq{eq:regret-naive} for the \naiveAlg. We obtain nearly matching lower bounds using the KL-divergence technique from~\citep{bandits-exp3,Bobby-nips04}. The lower bound is very general as it holds for every given (product) similarity space \emph{and} for every fixed value of the upper bound.

Our stochastic contextual MAB setting, and specifically the \ZoomAlg, can be fruitfully applied beyond the ad placement scenario described above and beyond MAB with similarity information per se. First, writing $x_t=t$ one can incorporate ``temporal constraints" (across time, for each arm), and combine them with ``spatial constraints" (across arms, for each time). The analysis of contextual zooming yields concrete, meaningful bounds this scenario. In particular, we recover one of the main results in~\cite{DynamicMAB-colt08}. Second, our setting subsumes the stochastic \emph{sleeping bandits} problem~\cite{sleeping-colt08}, where in each round some arms are ``asleep", i.e. not available in this round. Here contexts correspond to subsets of arms that are ``awake". Contextual zooming recovers and generalizes the corresponding result in~\cite{sleeping-colt08}. Third, following the publication of a preliminary version of this paper, contextual zooming has been applied to bandit learning-to-rank in~\cite{ZoomingRBA-icml10}.

\OMIT{
it applies to a version of the adversarial MAB problem in which an adversary is constrained to change the expected payoffs of each arm \emph{gradually}, e.g. by a small amount in each round. In fact, we can combine significant constraints across time (for each arm) and across arms (for each time). }

\OMIT{For the context-free setting, our guarantees match those in~\cite{LipschitzMAB-stoc08}.  Our algorithm and analysis extends to a more general setting where some context-arms pairs may be unfeasible, and moreover the right-hand side of~\refeq{eq:LipschitzD} is replaced by an arbitrary metric on the feasible context-arms pairs.}

\OMIT{We apply the \ZoomAlg{} to a (context-free) adversarial MAB problem in which an adversary is constrained to change the expected payoffs of each arm \emph{gradually}, e.g. by a small amount in each round. This setting is naturally modeled as a contextual MAB problem in which the $t$-th context arrival is simply $x_t=t$. Then $\mu(t,y)$ corresponds to the expected payoff of arm $y$ at time $t$, and the context metric $\D_\TX(t,t')$ provides an upper bound on the temporal change
	$|\mu(t,y)-\mu(t',y)|$.
We term it the {\bf\em \driftProblem}. Interestingly, this problem incorporates significant constraints both across time (for each arm) and across arms (for each time); to the best of our knowledge, such MAB models are quite rare in the literature.\footnote{The only other MAB model with this flavor that we are aware of, found in~\cite{Hazan-soda09}, combines linear payoffs and bounded ``total variation" (aggregate temporal change) of the cost functions.} Notable special cases of $\D_\TX$ include
	$\D_\TX(t,t') = \sigma|t-t'| $ and 	$\D_\TX(t,t') = \sigma \sqrt{|t-t'|} $,
which corresponds to, respectively, the bounded change per round and the high-probability behavior of a random walk. We derive provable guarantees for these two examples, and show that they are essentially optimal.

Interestingly, the \problem{} subsumes the stochastic \emph{sleeping bandits} problem~\cite{sleeping-colt08}, where in each round some arms are ``asleep", i.e. not available in this round. Each context arrival $x_t$ corresponds to the set of arms that are ``awake" in this round. More precisely, contexts $x$ correspond to subsets $S_x$ of arms, so that only the context-arm pairs $(x,y)$, $y\in S_x$ are feasible, and the context distance is
	$\D_X \triangleq 0$.
Moreover, the \problem{} extends the sleeping bandits setting by incorporating similarity information on arms. The \ZoomAlg{} (and its analysis) applies, and is geared to exploit this additional similarity information. In the absence of such information the algorithms essentially reduces to the ``highest awake index" algorithm in~\cite{sleeping-colt08}.
} %%%%%%%

\OMIT{ %%%%%%%
The analysis of \ZoomAlg{} carries over to the \driftProblem; contextual regret becomes \emph{dynamic regret} -- regret with respect to a benchmark which in each round plays the best arm for this round. In this setting, the quantity of interest is average dynamic regret, which is typically independent of the time horizon.
} %%%%%%%%%%%%%%

\OMIT{ %%%%%%%%%%%
Furthermore, we consider the {\em Dynamic MAB problem}~\cite{DynamicMAB-colt08} in which the \emph{state} (the current expected payoff) of each arm undergoes an independent Brownian motion on a $[0,1]$ interval with reflecting boundaries. We treat this problem as (essentially) a special case of the \driftProblem{} such that
	$\D_\TX(t,t') = \sigma \sqrt{|t-t'|} $,
where $\sigma$ is the \emph{volatility} (speed of change) of the Brownian motion. We improve the analysis of \ZoomAlg{} to obtain guarantees that are superior to those for the algorithms in~\cite{DynamicMAB-colt08}, and provide a nearly matching lower bound.
} %%%%%%%%

For the adversarial setting, we provide an algorithm which maintains an adaptive partition of the context space and thus takes advantage of ``benign" context arrivals. We develop provable guarantees that capture this ``benign-ness". In the worst case, the contextual regret is bounded in terms of the covering dimension of the context space, matching~\refeq{eq:regret-naive}. Our algorithm is in fact a \emph{meta-algorithm}: given an adversarial bandit algorithm \bandit, we present a contextual bandit algorithm which calls \bandit{} as a subroutine. Our setup is flexible: depending on what additional constraints are known about the adversarial payoffs, one can plug in a bandit algorithm from the prior work on the corresponding version of adversarial MAB, so that the regret bound for \bandit{} plugs into the overall regret bound.

\OMIT{Our setup allows us to leverage prior work on other adversarial MAB formulations, such as the basic $k$-arm version~\cite{bandits-exp3}, linear payoffs~\cite{McMahan-colt04,Bobby-stoc04,DaniHK-nips07,AbernethyHR-colt08,Hazan-soda09} and convex payoffs~\cite{Bobby-nips04,FlaxmanKM-soda05}.}

%%%%%%%%%%%%%%%%%%%%%%%%%%%%%%%%%%%%%
\xhdr{Discussion.} Adaptive partitions (of the arms space) for context-free MAB with similarity information have been introduced in~\citep{LipschitzMAB-stoc08,xbandits-nips08}. This paper further explores the potential of the zooming technique in~\citep{LipschitzMAB-stoc08}. Specifically, contextual zooming extends this technique to adaptive partitions of the entire similarity space, which necessitates a technically different algorithm and a more delicate analysis. We obtain a clean algorithm for contextual MAB with improved (and nearly optimal) bounds. Moreover, this algorithm applies to several other, seemingly unrelated problems and unifies some results from prior work.

One alternative approach is to maintain a partition of the context space, and run a separate instance of the zooming algorithm from~\citet{LipschitzMAB-stoc08} on each set in this partition. Fleshing out this idea leads to the meta-algorithm that we present for adversarial payoffs (with \bandit{} being the zooming algorithm). This meta-algorithm is parameterized (and constrained) by a specific a priori regret bound for \bandit. Unfortunately, any a priori regret bound for zooming algorithm would be a pessimistic one, which negates its main strength -- the ability to adapt to ``benign" expected payoffs.

%%%%%%%%%%%%%%%%%%%%%%%%%%%%%%%%%%%%%
\xhdr{Map of the paper.}
Section~\ref{sec:related-work} is related work, and Section~\ref{sec:prelims} is Preliminaries. Contextual zooming is  presented in Section~\ref{sec:ZoomAlg}. Lower bounds are in Section~\ref{sec:LBs}. Some applications of contextual zooming are discussed in Section~\ref{sec:apps}.  The adversarial setting is treated in Section~\ref{sec:CovAlg}.
%We conclude in Section~\ref{sec:conclusions}.

\newpage
\section{Related work}
\label{sec:related-work}

%~\citep{sundaram-bandits-92,Agrawal-bandits-95,bandits-exp3,Bobby-stoc04,Bobby-nips04,McMahan-colt04,Bobby-thesis,FlaxmanKM-soda05,Cope/06,Hayes-soda06,DaniHK-nips07,AuerOS/07,KakadeKL/07,LipschitzMAB-stoc08,xbandits-nips08,DichotomyMAB-soda10,Bubeck-colt10,Munos-ecml10}.

% the hack for put a footnote exactly where I want (part I)
\newcounter{FnIndWork}
\addtocounter{footnote}{1}
\setcounter{FnIndWork}{\value{footnote}}
\newcommand{\IndWork}{\footnotemark[\value{FnIndWork}]}

A proper discussion of the literature on bandit problems is beyond the scope of this paper. This paper follows the line of work on regret-minimizing bandits;
a reader is encouraged to refer to~\citep{CesaBL-book,Bubeck-survey12} for background. A different (Bayesian) perspective on bandit problems can be found in \citep{Gittins-book11}.

Most relevant to this paper is the work on bandits with large sets of arms, specifically bandits with similarity information~\citep{Agrawal-bandits-95,Bobby-nips04,AuerOS/07,yahoo-bandits07,Kocsis-ecml06,Munos-uai07,LipschitzMAB-stoc08,xbandits-nips08,DichotomyMAB-soda10,Munos-ecml10}.
Another commonly assumed structure is linear or convex payoffs, e.g. ~\citep{Bobby-stoc04,FlaxmanKM-soda05,DaniHK-nips07,AbernethyHR-colt08,Hazan-soda09,bubeck-colt12}.
Linear/convex payoffs is a much stronger assumption than similarity, essentially because it  allows to make strong inferences about far-away arms. Other assumptions have been considered, e.g. %\citep{sundaram-bandits-92,Berry-bandits-97,Munos-nips08,Bubeck-colt10}.
\citep{Munos-nips08,Bubeck-colt10}. The distinction between stochastic and adversarial payoffs is orthogonal to the structural assumption (such as Lipschitz-continuity or linearity). Papers on MAB with linear/convex payoffs typically allow adversarial payoffs, whereas papers on MAB with similarity information focus on stochastic payoffs, with notable exceptions of~\citet{Bobby-nips04} and~\citet{Munos-ecml10}.\IndWork

The notion of structured adversarial payoffs in this paper is less restrictive than the one in~\citet{Munos-ecml10} (which in turn specializes the notion from linear/convex payoffs), in the sense that the Lipschitz condition is assumed on the expected payoffs rather than on realized payoffs. This is a non-trivial distinction, essentially because our notion generalizes stochastic payoffs whereas the other one does not.

\OMIT{In particular,~\citet{Munos-ecml10} achieve regret $\tilde{O}(\sqrt{d T})$ for $d$-dimensional real space, whereas (even) for stochastic payoffs there is a lower bound $\Omega(T^{1-1/(d+2)})$~\citep{Bobby-nips04,xbandits-nips08}.}

% the hack for put this footnote exactly where I want (part II)
\footnotetext[\value{footnote}]{This paper is concurrent and independent work w.r.t. the preliminary publication of this paper on {\tt arxiv.org}.}

\xhdr{Contextual MAB.} In \citep{Auer-focs00} and~\citep{Reyzin-aistats11-linear}\IndWork{} payoffs are linear in context, which is a feature vector.~\citep{Woodroofe79,Wang-sideMAB05} and~\citep{Zeevi-colt10}\IndWork{} study contextual MAB with stochastic payoffs, under the name \emph{bandits with covariates}: the context is a random variable correlated with the payoffs; they consider the case of two arms, and make some additional assumptions. \cite{Lazaric-colt09}\IndWork{} consider an online labeling problem with stochastic inputs and adversarially chosen labels; inputs and hypotheses (mappings from inputs to labels) can be thought of as ``contexts" and ``arms" respectively. \emph{Bandits with experts advice} (e.g. \citet{Auer-focs00}) is the special case of contextual MAB where the context consists of experts' advice; the advice of a each expert is modeled as a distributions over arms. All these papers are not directly applicable to the present setting.

Experimental work on contextual MAB includes~\citep{yahoo-bandits07} and \citep{Langford-www10,Langford-wsdm11}.\IndWork

\cite{Pal-Bandits-aistats10}\IndWork{} consider the setting in this paper for a product similarity space and, essentially, recover the \naiveAlg{} and a lower bound that matches~\refeq{eq:regret-naive}. The same guarantee~\refeq{eq:regret-naive} can also be obtained as follows. The ``uniform partition" described above can be used to define ``experts" for a bandit-with-expert-advice algorithm such as \EXP[4]~\citep{bandits-exp3}: for each set of the partition there is an expert whose advise is simply an arbitrary arm in this set. Then the regret bound for \EXP[4]  yields~\refeq{eq:regret-naive}. Instead of \EXP[4] one could use an algorithm in~\cite{McMahan-colt09}\IndWork{} which improves over \EXP[4] if the experts are not ``too distinct"; however, it is not clear if it translates into concrete improvements over~\refeq{eq:regret-naive}.

If the context $x_t$ is time-invariant, our setting reduces to
the Lipschitz MAB problem as defined in~\citep{LipschitzMAB-stoc08}, which in turn reduces to continuum-armed bandits~\citep{Agrawal-bandits-95,Bobby-nips04,AuerOS/07} if the metric space is a real line, and to MAB with stochastic payoffs~\citep{bandits-ucb1} if the similarity information is absent.

%%%%%%%%%%%%
\section{Preliminaries}
\label{sec:prelims}

\newcommand{\Npack}{N^{\mathtt{pack}}} % r-packing number, without the subscript r.

We will use the notation from the Introduction. In particular, $x_t$ will denote the $t$-th \emph{context arrival}, i.e. the context that arrives in round $t$, and $y_t$ will denote the arm chosen by the algorithm in that round. We will use $\arr$ to denote the sequence of the first $T$ context arrivals $(x_1 \LDOTS x_T)$. The \emph{badness} of a point $(x,y)\in \mP$ is defined as
    $\Delta(x,y) \triangleq  \mu^*(x) - \mu(x,y)$.
The context-specific best arm is
\begin{align}\label{eq:bechmark-defn-stochastic}
y^*(x) \in \textstyle{\argmax_{y\in Y:\, (x,y)\in\mP}}\; \mu(x, y),
\end{align}
where ties are broken in an arbitrary but fixed way. To ensure that
the $\max$ in~\refeq{eq:bechmark-defn-stochastic} is attained by some $y\in Y$, we will assume that the similarity space $(\mP,\D)$ is compact.

\OMIT{ %%%%%%%%% ADVERSARIAL LIPSCHITZ CONDITIONS
The similarity information is given to an algorithm as two metric spaces $(X,\D_\TX)$ and $(Y,\D_\TY)$ called, respectively, the \emph{context space} and the \emph{arms space}, such that the following Lipschitz-style conditions hold:
\begin{align}
|\mu_t(x,y) - \mu_t(x',y')|
	&\leq \D_\TX(x,x') + \D_\TY(y,y'), \label{eq:Lipschitz} \\
|\mu^*_t(x) - \mu^*_t(x')|
	&\leq \D_\TX(x,x')
	\label{eq:Lipschitz-star}
\end{align}
} %%%%%%%

{\bf Metric spaces.}
Covering dimension and related notions are crucial throughout this paper. Let $\mP$ be a set of points in a metric space, and fix $r>0$. An \emph{$r$-covering} of $\mP$ is a collection of subsets of $\mP$, each of diameter strictly less than $r$, that cover $\mP$. The minimal number of subsets in an $r$-covering is called the \emph{$r$-covering number} of $\mP$ and denoted $N_r(\mP)$.~\footnote{The covering number can be defined via radius-$r$ balls rather than diameter-$r$ sets. This alternative definition lacks the appealing ``robustness" property: $N_r(\mP')\leq N_r(\mP)$ for any $\mP'\subset \mP$, but (other than that) is equivalent for this paper.} The \emph{covering dimension} of $\mP$ (with multiplier $c$) is the smallest $d$ such that
    $N_r(\mP) \leq c\, r^{-d}$
for each $r>0$. In particular, if $S$ is a subset of Euclidean space then its covering dimension is at most the linear dimension of $S$, but can be (much) smaller.

\OMIT{The following ``robustness" property holds: $N_r(\mP')\leq N_r(\mP)$ for any $\mP'\subset \mP$, and similarly for the covering dimension.}

Covering is closely related to packing. A subset $S\subset \mP$ is an \emph{$r$-packing} of $\mP$ if the distance between any two points in $S$ is at least $r$. The maximal number of points in an $r$-packing is called the \emph{$r$-packing number} and denoted $\Npack_r(\mP)$. It is well-known that $r$-packing numbers are essentially the same as $r$-covering numbers, namely
    $N_{2r}(\mP) \leq \Npack_r(\mP) \leq N_r(\mP)$.

\OMIT{The \emph{doubling constant} of $\mP$, denoted $\Cdbl(\mP)$, is the smallest constant $c>0$ such that
    $N_{\delta\,r}(S) \leq c^{-\log \delta}$
for any $\delta\in(0,1)$ and any subset $S\subset \mP$ of diameter $r$.
The covering dimension (with multiplier $1$) is at most $\log \Cdbl$, but can be much smaller in general. The following fact is well-known: if distance between any two points in $S$ is $>r$, then any ball of radius $r$ contains at most $\Cdbl^2$ points of $S$.}

The \emph{doubling constant} $\Cdbl(\mP)$ of $\mP$ is the smallest $k$ such that any ball can be covered by $k$ balls of half the radius. The doubling constant (and \emph{doubling dimension} $\log \Cdbl$) was introduced in~\cite{Heinonen01} and has been a standard notion in theoretical computer science literature since~\citet{Gup03}. It was used to characterize tractable problem instances for a variety of problems~\citep[e.g. see][]{Tal04,Slivkins-focs04,Cole-stoc06}. It is known that $\Cdbl(\mP)\geq c\,2^d$ if $d$ is the covering dimension of $\mP$ with multiplier $c$, and that $\Cdbl(\mP)\leq 2^d$ if $\mP$ is a bounded subset of $d$-dimensional Euclidean space. A useful observation is that if distance between any two points in $S$ is $>r$, then any ball of radius $r$ contains at most $\Cdbl$ points of $S$.

A ball with center $x$ and radius $r$ is denoted $B(x,r)$. Formally, we will treat a ball as a (center, radius) pair rather than a set of points. A function $f:\mP\to \R$ if a Lipschitz function on a metric space $(\mP,\D)$, with Lipschitz constant $\LipConst$, if the \emph{Lipschitz condition} holds:
	$|f(x)-f(x')|\leq \LipConst\,\D(x,x')$
for each $x,x'\in \mP$.

{\bf Accessing the similarity space.}
We assume full and computationally unrestricted access to the similarity information. While the issues of efficient representation thereof are important in practice, we believe that a proper treatment of these issues would be specific to the particular application and the particular similarity metric used, and would obscure the present paper. One clean formal way to address this issue is to assume \emph{oracle access}: an algorithm accesses the similarity space via a few specific types of queries, and invokes an ``oracle" that answers such queries.

{\bf Time horizon.} We assume that the time horizon is fixed and known in advance. This assumption is without loss of generality in our setting. This is due to the well-known \emph{doubling trick} which converts a bandit algorithm with a fixed time horizon into one that runs indefinitely and achieves essentially the same regret bound. Suppose for any fixed time horizon $T$ there is an
algorithm $\mathtt{ALG}_T$ whose regret is at most $R(T)$. The new algorithm proceeds in phases $i= 1,2,3,\ldots$ of duration $2^i$ rounds each, so that in each phase $i$ a fresh instance of $\mathtt{ALG}_{2^i}$ is run. This algorithm has regret $O(\log T) R(T)$ for each round $T$, and $O(R(T))$ in the typical case when $R(T)\geq T^{\gamma}$ for some constant $\gamma>0$.

\OMIT{An \emph{$r$-net} is the set $S$ of points in a metric space such that any two points in $S$ are at distance $> r$ from each other, and each point in the metric space is within distance $\leq r$ from some point in $S$. }

\OMIT{
The algorithms in this paper are formulated for a fixed time horizon, i.e. they provide a regret bound for one fixed time $T$ which is given as a parameter. There is a simple way to transform a fixed time horizon algorithm $\A(T)$ with regret bound of the form
    $R(T) \leq O(T^\gamma)$
into the an algorithm which ``works" indefinitely and has the same regret bound (up to a constant factor). Namely, the \emph{doubling trick}: partition time in phases of exponential duration; in each phase $i$ of duration $T = 2^i$ run a fresh copy of algorithm $\A(2^i)$.
}

%%%%%%%%%%%%
\section{The \ZoomAlg}
\label{sec:ZoomAlg}

In this section we consider the \problem{} with stochastic payoffs. We present an algorithm for this problem, called \emph{contextual zooming}, which takes advantage of both the ``benign" context arrivals and the  ``benign" expected payoffs. The algorithm adaptively maintains a partition of the similarity space, ``zooming in" on both the ``popular" regions on the context space and the high-payoff regions of the arms space.

Contextual zooming extends the (context-free) zooming technique in~\citep{LipschitzMAB-stoc08}, which necessitates a somewhat more complicated algorithm. In particular, selection and activation rules are defined differently, there is a new notion of ``domains" and the distinction between ``pre-index" and ``index". The analysis is more delicate, both the high-probability argument in Claim~\ref{cl:contZoom-chernoff} and the subsequent argument that bounds the number of samples from suboptimal arms. Also, the key step of setting up the regret bounds is very different, especially for the improved regret bounds in Section~\ref{subsec:stochastic-improved}.

%%%%%%%%%%%%
\subsection{Provable guarantees}
\label{sec:stochastic-guarantees}

Let us define the notions that express the performance of contextual zooming. These notions rely on the packing number $N_r(\cdot)$ in the similarity space $(\mP,\D)$, and the more refined versions thereof that take into account ``benign" expected payoffs and ``benign" context arrivals.

Our guarantees have the following form, for some integer numbers
    $\{ N_r \}_{r\in (0,1)}$:
\begin{align}\label{eq:regret-covNum}
R(T) 	\leq C_0\; \textstyle{\inf_{r_0\in (0,1)}}
		\left(
			r_0 T	+
            \textstyle{\sum_{r=2^{-i}:\; i\in\N,\,r_0\leq r\leq 1}}\;\;
                \tfrac{1}{r}\,N_r \log T
		\right).
\end{align}
Here and thereafter, $C_0=O(1)$ unless specified otherwise. In the pessimistic version, $N_r = N_r(\mP)$ is the $r$-packing number of $\mP$.
\footnote{Then~\refeq{eq:regret-covNum} can be simplified to
    $R(T)\leq \textstyle{\inf_{r\in (0,1)}}\,
    		O\left( r T	+ \tfrac{1}{r}\,N_r(\mP) \log T \right)$
since $N_r(\mP)$ is non-increasing in $r$.}
The main contribution is refined bounds in which $N_r$ is smaller.

For every guarantee of the form~\refeq{eq:regret-covNum}, call it \emph{$N_r$-type} guarantee, prior work (e.g.,~\cite{Bobby-nips04,LipschitzMAB-stoc08,xbandits-nips08}) suggests a more tractable \emph{dimension-type} guarantee. This guarantee is in terms of the \emph{covering-type dimension} induced by $N_r$, defined as follows:\footnote{One standard definition of the covering dimension is~\refeq{eq:def-dim} for $N_r = N_r(\mP)$ and $c=1$. Following~\cite{LipschitzMAB-stoc08}, we include an explicit dependence on $c$ in~\refeq{eq:def-dim} to obtain a more efficient regret bound (which holds for any $c$).}
\begin{align}\label{eq:def-dim}
d_c \triangleq \inf\{ d>0:\; N_r \leq c\, r^{-d} \quad \forall r\in(0,1) \}.
\end{align}
Using~\refeq{eq:regret-covNum} with
    $r_0 = T^{-1/(d_c+2)}$,
we obtain
\begin{align}\label{eq:regret-dim}
R(T) \leq O(C_0)\,(c\,T^{1-1/(2+d_c)}\, \log T)
    \qquad (\forall c>0).
\end{align}

For the pessimistic version ($N_r = N_r(\mP)$), the corresponding covering-type dimension $d_c$ is the covering dimension of the similarity space. The resulting guarantee~\refeq{eq:regret-dim} subsumes the bound~\refeq{eq:regret-naive} from prior work (because the covering dimension of a product similarity space is $d_\TX+d_\TY$), and extends this bound from product similarity spaces~\refeq{eq:product-space} to arbitrary similarity spaces.

To account for ``benign" expected payoffs, instead of $r$-packing number of the entire set $\mP$ we consider the $r$-packing number of a subset of $\mP$ which only includes points with near-optimal expected payoffs:
\begin{align}\label{eq:P-mu-r}
\mPzoom
    \triangleq \{(x,y)\in \mP:\, \mu^*(x) - \mu(x,y) \leq 12\,r\}.
\end{align}
We define the \emph{$r$-zooming number} as
    $N_r(\mPzoom)$,
the $r$-packing number of $\mPzoom$. The corresponding covering-type dimension~\refeq{eq:def-dim} is called the \emph{\ZoomDim}.

The $r$-zooming number can be seen as an optimistic version of $N_r(\mP)$: while equal to $N_r(\mP)$ in the worst case, it can be much smaller if the set of near-optimal context-arm pairs is ``small" in terms of the packing number. Likewise, the \ZoomDim{} is an optimistic version of the covering dimension.

\begin{theorem}\label{thm:ZoomAlg}
Consider the \problem{} with stochastic payoffs. There is an algorithm (namely, Algorithm~\ref{alg:contextual-zooming} described below) whose contextual regret $R(T)$ satisfies~\refeq{eq:regret-covNum} with $N_r$ equal to $N_r(\mPzoom)$, the $r$-zooming number. Consequently, $R(T)$ satisfies the dimension-type guarantee~\refeq{eq:regret-dim}, where $d_c$ is the \ZoomDim.
\end{theorem}

%\begin{note}{Remark.}
In Theorem~\ref{thm:ZoomAlg}, the same algorithm enjoys the bound~\refeq{eq:regret-dim} for each $c>0$. This is  a useful trade-off since different values of $c$ may result in drastically different values of the dimension $d_c$. On the contrary, the ``\naiveAlg" from prior work essentially needs to take the $c$ as input.
%\end{note}

Further refinements to take into account ``benign" context arrivals are deferred to Section~\ref{subsec:stochastic-improved}.

\newcommand{\rad}{\mathtt{conf}} % confidence radius
\newcommand{\domain}{\mathtt{dom\,}} % domain
\newcommand{\preindex}{I^{\text{pre}}} % pre-index
\newcommand{\parent}{B^{\text{par}}} % parent ball
\newcommand{\B}{\mathcal{B}}

%%%%%%%%%%%%%%%%%%%%%%%%%%
\subsection{Description of the algorithm}
\label{subsec:stochastic-algorithm}

The algorithm is parameterized by the time horizon $T$. In each round $t$, it maintains a finite collection $\mA_t$ of balls in $(\mP,\D)$ (called \emph{active balls}) which collectively cover the similarity space. Adding active balls is called \emph{activating}; balls stay active once they are activated. Initially there is only one active ball which has radius $1$ and therefore contains the entire similarity space.

At a high level, each round $t$ proceeds as follows. Context $x_t$ arrives. Then the algorithm selects an active ball $B$ and an arm $y_t$ such that $(x_t,y_t)\in B$, according to the ``selection rule". Arm $y_t$ is played. Then one ball may be activated, according to the ``activation rule".

In order to state the two rules, we need to put forward several definitions. Fix an active ball $B$ and round $t$. Let $r(B)$ be the radius of $B$.  The \emph{confidence radius} of $B$ at time $t$ is
\begin{align}\label{eq:def-confRad}
\rad_t(B)
%    \triangleq \rad(n_t(B))
    \triangleq 4\, \sqrt{\frac{\log T}{1+n_t(B)}},
\end{align}
where $n_t(B)$ is the number of times $B$ has been selected by the algorithm before round $t$. The \emph{domain} of ball $B$ in round $t$ is a subset of $B$ that excludes all balls $B'\in \mA_t$ of strictly smaller radius:
\begin{align}\label{eq:def-domain}
\domain_t(B)
    \triangleq B \setminus
    \left( \textstyle \bigcup_{B'\in \mA_t:\, r(B')< r(B)}\; B' \right).
\end{align}
We will also denote \eqref{eq:def-domain} as $\domain(B,\mA_t)$.
Ball $B$ is called \emph{relevant} in round $t$ if
    $(x_t,y)\in \domain_t(B)$
for some arm $y$. In each round, the algorithm selects one relevant ball $B$. This ball is selected according to a numerical score $I_t(B)$ called \emph{index}. (The definition of index is deferred to the end of this subsection.)

Now we are ready to state the two rules, for every given round $t$.
\begin{itemize}
\item {\bf\em selection rule.} Select a relevant ball $B$ with the maximal index (break ties arbitrarily). Select an arbitrary arm $y$ such that
        $(x_t,y)\in \domain_t(B)$.
\item {\bf\em activation rule.} Suppose the selection rule selects a relevant ball $B$ such that
        $\rad_t(B)\leq r(B)$
    after this round.
    Then, letting $y$ be the arm selected in this round, a ball with center $(x_t, y)$ and radius $\tfrac12\, r(B)$ is activated. ($B$ is then called the \emph{parent} of this ball.)
\end{itemize}
See Algorithm~\ref{alg:contextual-zooming} for the pseudocode.

\begin{algorithm}[t]
\begin{algorithmic}[1]
\caption{Contextual zooming algorithm.}
\label{alg:contextual-zooming}

\STATE {\bf Input:}
    Similarity space $(\mP,\D)$ of diameter $\leq 1$, $\mP\subset X\times Y$.
    Time horizon $T$.
\STATE {\bf Data:}
    collection $\mA$ of ``active balls" in $(\mP,\D)$;~~
    counters $n(B)$, $\reward(B)$ for each $B\in \mA$.
\STATE {\bf Init:}
    $B\leftarrow B(p, 1)$;~~
    $\mA \leftarrow \{B\}$;~~
    $n(B)=\reward(B) = 0$
\COMMENT{center $p\in \mP$ is arbitrary}
\STATE {\bf Main loop:} for each round $t$
    \COMMENT{use definitions (\ref{eq:def-confRad}-\ref{eq:def-index})}
\STATE \TAB Input context $x_t$.
\STATE \TAB \COMMENT{activation rule}
\STATE \TAB $\RELEVANT \leftarrow
        \{ B\in \mA: (x_t,y)\in \domain(B,\mA) \text{ for some arm $y$}\} $.
\STATE \TAB $B\leftarrow \argmax_{B\in\RELEVANT}\; I_t(B)$.
    \COMMENT{ball $B$ is selected}
\STATE \TAB $y\leftarrow \text{ any arm y such that $(x_t,y)\in \domain(B,\mA)$}$.
\STATE \TAB Play arm $y$, observe payoff $\pi$.
\STATE \TAB Update counters: $n(B)\leftarrow n(B)+1$,
                $\reward(B)\leftarrow \reward(B)+\pi$.
\STATE \TAB \COMMENT{selection rule}
\STATE \TAB   {\bf if} $\rad(B)\leq \mathtt{radius}(B)$ {\bf then}
\STATE \TAB\TAB
    $B'\leftarrow B((x_t,y),\; \tfrac12\,\mathtt{radius}(B))$
    \COMMENT{new ball to be activated}
\STATE \TAB\TAB
    $\mA\leftarrow\mA \cup \{ B'\}$;~~
    $n(B')=\reward(B')=0$.
\end{algorithmic}
\end{algorithm}

It remains to define the index $I_t(B)$. Let $\reward_t(B)$ be the total payoff from all rounds up to $t-1$ in which ball $B$ has been selected by the algorithm. Then the average payoff from $B$ is
%\begin{align}
$\nu_t(B) \triangleq \tfrac{\reward_t(B)}{\max(1,\;n_t(B))}$.
%\end{align}
The \emph{pre-index} of $B$ is defined as the average $\nu_t(B)$ plus an ``uncertainty term":
\begin{align}\label{eq:def-preindex}
\preindex_t(B)
	\triangleq \nu_t(B) + r(B) + \rad_t(B).
\end{align}

\noindent The ``uncertainty term" in~\refeq{eq:def-preindex} reflects both uncertainty due to a location in the metric space, via $r(B)$, and uncertainty due to an insufficient number of samples, via $\rad_t(B)$.

The index of $B$ is obtained by taking a minimum over all active balls $B'$:
\begin{align}\label{eq:def-index}
I_t(B)
	\triangleq r(B)+ \min_{B'\in \mA_t}
        \; \left( \preindex_t(B') + \D(B,B') \right),
\end{align}
where $\D(B,B')$ is the distance between the centers of the two balls.

\xhdr{Discussion.} The meaning of index and pre-index is as follows. Both are upper confidence bound (\emph{UCB}, for short) for expected rewards in $B$. Pre-index is a UCB for $\mu(B)$, the expected payoff from the center of $B$; essentially, it is the best UCB on $\mu(B)$ that can be obtained from the observations of $B$ alone. The $\min$ expression in \eqref{eq:def-index} is an improved UCB on $\mu(B)$, refined using observations from all other active balls. Finally, index is, essentially, the best available UCB for the expected reward of any pair $(x,y)\in B$.

Relevant balls are defined through the notion of the ``domain" to ensure the following property: in each round when a parent ball is selected, some other ball is activated. This property allows us to ``charge" the regret accumulated in each such round to the corresponding activated ball.

\xhdr{Running time.}
The running time is dominated by determining which active balls are relevant. Formally, we assume an oracle that inputs context $x$ and a finite sequence $(B, B_1 \LDOTS B_n)$ of balls in the similarity space, and outputs an arm $y$ such that
    $(x,y)\in B \setminus \cup_{j=1}^n B_j$
if such arm exists, and {\tt null} otherwise. Then each round $t$ can be implemented via $n_t$ oracle calls with $n<n_t$ balls each, where $n_t$ is the current number of active balls. Letting $f(n)$ denote the running time of one oracle call in terms of $n$, the running time for each round  the algorithm is at most $n_T\, f(n_T)$.

While implementation of the oracle and running time $f(\cdot)$ depend on the specific similarity space, we can provide some upper bounds on $n_T$. First, a crude upper bound is $n_T\leq T$. Second, letting $\F_r$ be the collection of all active balls of radius $r$, we prove that $|\F_r|$ is at most $N_r$, the $r$-zooming number of the problem instance. Third, $|\F_r|\leq \Cdbl\, Tr^2$, where $\Cdbl$ is the doubling constant of the similarity space. (This is because each active ball must be played at least $r^{-2}$ times before it becomes a parent ball, and each parent ball can have at most $\Cdbl$ children.) Putting this together, we obtain
    $n_T \leq \sum_r \min(\Cdbl\, Tr^2,\, N_r)$,
where the sum is over all $r = 2^{-j}$, $j\in \N$.

% from arxiv'11 version
%\begin{align}\label{eq:def-index}
%I_t(B)
%	\triangleq \min_{B'\in \mA_t:\; r(B')\geq r(B) }
%        \; \preindex_t(B') + \D(B,B').
%\end{align}

\subsection{Analysis of the algorithm: proof of Theorem~\ref{thm:ZoomAlg}}
\label{subsec:stochastic-analysis}

We start by observing that the activation rule ensures several important invariants.

\begin{claim}\label{eq:contextZoom-invariants}
The following invariants are maintained:
\begin{OneLiners}
%\item (centering) if $B$ is activated in round $t$ with parent $\parent$,
% then the center of $B$ is $(x_t,y_t)\in \domain(\parent,\mA)$.
\item (confidence)
for all times $t$ and all active balls $B$,
\begin{align*}
   \text{ $\rad_t(B)\leq r(B)$ if and only if $B$ is a parent ball.}
\end{align*}
    \OMIT{$\rad_t(B)> r(B)$ for all active balls $B$ and all rounds $t$.}
\item (covering) in each round $t$, the domains of active balls cover the similarity space.

\item (separation) for any two active balls of radius $r$, their centers are at distance at least $r$.
\end{OneLiners}
\end{claim}

\begin{proof}
The confidence invariant is immediate from the activation rule.

For the covering invariant, note that
$\cup_{B\in \mA}\, \domain(B,\mA) = \cup_{B\in \mA}\, B$
for any finite collection $\mA$ of balls in the similarity space. (For each $v\in \cup_{B\in \mA}\,B$, consider a  smallest radius ball in $\mA$ that contains $B$. Then $v\in \domain(B,\mA)$.) The covering invariant then follows since $\mA_t$ contains a ball that covers the entire similarity space.

To show the separation invariant, let $B$ and $B'$ be two balls of radius $r$ such that $B$ is activated at time $t$, with parent $\parent$, and $B'$ is activated before time $t$. The center of $B$ is some point
    $(x_t,y_t) \in \domain(\parent,\mA_t)$.
Since $r(\parent)>r(B')$, it follows that
    $(x_t,y_t) \not\in B'$.
\end{proof}

%%%%%%%%%%%%%%%%%%%%%%%%%

\newcommand{\Bsel}{B_t^{\mathtt{sel}}}

Throughout the analysis we will use the following notation.
For a ball $B$ with center $(x,y)\in \mP$, define the expected payoff of $B$ as
    $\mu(B) \triangleq \mu(x,y)$.
Let $\Bsel$ be the active ball selected by the algorithm in round $t$. Recall that the \emph{badness} of $(x,y)\in \mP$ is defined as
    $\Delta(x,y) \triangleq  \mu^*(x) - \mu(x,y)$.

\begin{claim}\label{cl:contZoom-chernoff}
If ball $B$ is active in round $t$, then with probability at least $1-T^{-2}$ we have that
\begin{align}\label{eq:chernoff-contextual}
    |\nu_t(B) - \mu(B)| \leq r(B)+ \rad_t(B).
\end{align}

\end{claim}

\begin{proof}
Fix ball $V$ with center $(x,y)$. Let $S$ be the set of rounds $s\leq t$ when ball $B$ was selected by the algorithm, and let $n = |S|$ be the number of such rounds. Then
$ \nu_t(B) = \tfrac{1}{n}\, \textstyle{\sum_{s\in S}}\;
                    \pi_s(x_s, y_s)
$.

Define
%\begin{align*}
$Z_k = \sum \left( \pi_s(x_s, y_s) - \mu(x_s, y_s) \right)$,
%\nu'_t(B)       &= \tfrac{1}{n}\, \textstyle{\sum_{s\in S}}\; \pi'_s.
%\end{align*}
where the sum is taken over the $k$ smallest elements $s\in S$. Then
    $\{Z_{k\wedge n}\}_{k\in \N}$
is a martingale with bounded increments. (Note that $n$ here is a random variable.) So by the Azuma-Hoeffding inequality with probability at least $1-T^{-3}$ it holds that
	$\tfrac{1}{k}\,| Z_{k\wedge n}| \leq \rad_t(B)$,
for each $k\leq T$. Taking the Union Bound, it follows that
	$\tfrac{1}{n}\,| Z_n| \leq \rad_t(B)$.
Note that
	$|\mu(x_s, y_s)- \mu(B)| \leq r(B)$
for each $s\in S$, so
	$|\nu_t(B) - \mu(B)| \leq r(B) + \tfrac{1}{n}\,| Z_n|  $,
which completes the proof.
\end{proof}

Note that~\refeq{eq:chernoff-contextual} implies
    $\preindex(B)\geq \mu(B)$,
so that $\preindex(B)$ is indeed a UCB on $\mu(B)$.

Call a run of the algorithm \emph{clean} if~\refeq{eq:chernoff-contextual} holds for each round.
From now on we will focus on a clean run, and argue deterministically using~\refeq{eq:chernoff-contextual}. The heart of the analysis is the following lemma.

\begin{lemma}\label{lm:ZoomAlg-crux}
Consider a clean run of the algorithm. Then
	$\Delta(x_t,y_t) \leq 14\,r(\Bsel)$
in each round $t$.
\end{lemma}

\begin{proof}
Fix round $t$. By the covering invariant,
    $(x_t,\, y^*(x_t)) \in B$
for some active ball $B$. Recall from~\refeq{eq:def-index} that
	$I_t(B) = r(B) + \preindex(B') + \D(B,B')$
for some active ball $B'$.
Therefore
\begin{align}
I_t(\Bsel)
	&\geq I_t(B)
	=  \preindex(B') + r(B)+ \D(B,B')
		&\text{(selection rule, defn of index~\refeq{eq:def-index})} \nonumber\\
	&\geq \mu(B') + r(B) + \D(B,B')
		&\text{(``clean run")} \nonumber\\
	&\geq \mu(B) + r(B)
	\geq \mu(x_t,\, y^*(x_t)) =  \mu^*(x_t).
		&\text{(Lipschitz property~\refeq{eq:LipschitzD}, twice)}
	\label{eq:zoomAlg-crux-LB}
\end{align}
On the other hand, letting $\parent$ be the parent of $\Bsel$ and noting that by the activation rule
\begin{align}\label{eq:zoomAlg-parent}
\max(\D(\Bsel,\parent),\; \rad_t(\parent)) \leq r(\parent),
\end{align}
we can upper-bound $I_t(\Bsel)$ as follows:
\begin{align}
\preindex(\parent)
    &= \nu_t(\parent) + r(\parent) + \rad_t(\parent)
		&\text{(defn of preindex~\refeq{eq:def-preindex})} \nonumber\\	
    &\leq \mu(\parent) + 2\,r(\parent) + 2\,\rad_t(\parent)
        &\text{(``clean run")}\nonumber\\
    &\leq \mu(\parent) + 4\, r(\parent)
    			&\text{(``parenthood" \refeq{eq:zoomAlg-parent})} \nonumber \\
    &\leq \mu(\Bsel) + 5\, r(\parent)
        &\text{(Lipschitz property~\refeq{eq:LipschitzD})}
    		\label{eq:zoomAlg-crux-UB-1}\\
I_t(\Bsel)
	&\leq r(\Bsel) + \preindex(\parent)  + \D(\Bsel,\parent)
		&\text{(defn of index~\refeq{eq:def-index})} \nonumber\\
	&\leq r(\Bsel) + \preindex(\parent)  + r(\parent)
			&\text{(``parenthood" \refeq{eq:zoomAlg-parent})} \nonumber\\
	&\leq r(\Bsel) + \mu(\Bsel) + 6\,r(\parent)
		&\text{(by \refeq{eq:zoomAlg-crux-UB-1})}\nonumber\\
	&\leq \mu(\Bsel) + 13\,r(\Bsel)
		&\text{($r(\parent) = 2\,r(\Bsel)$)}\nonumber\\
	&\leq \mu(x_t,y_t) + 14\,r(\Bsel)
		&\text{(Lipschitz property~\refeq{eq:LipschitzD}).}
		\label{eq:zoomAlg-crux-UB}
\end{align}
Putting the pieces together,
	$\mu^*(x_t) \leq I_t(\Bsel) \leq \mu(x_t,y_t) + 14\,r(\Bsel)$.
\end{proof}

\begin{corollary}\label{cor:ZoomAlg-crux}
In a clean run, if ball $B$ is activated in round $t$ then
	$\Delta(x_t,y_t) \leq 10\,r(B)$.
\end{corollary}

\begin{proof}
By the activation rule, $\Bsel$ is the parent of $B$. Thus by Lemma~\ref{lm:ZoomAlg-crux} we immediately have
	$\Delta(x_t,y_t)\leq 14\,r(\Bsel) =28\, r(B)$.

To obtain the constant of 10 that is claimed here, we prove a more efficient special case of Lemma~\ref{lm:ZoomAlg-crux}:
\begin{align}\label{eq:lm:ZoomAlg-crux:more}
\text{if $\Bsel$ is a parent ball then
		$\Delta(x_t,y_t) \leq 5\,r(\Bsel)$.}
\end{align}

To prove~\refeq{eq:lm:ZoomAlg-crux:more}, we simply replace~\refeq{eq:zoomAlg-crux-UB} in the proof of Lemma~\ref{lm:ZoomAlg-crux} by similar inequality in terms of $\preindex(\Bsel)$ rather than $\preindex(\parent)$:
\begin{align*}
I_t(\Bsel)
	&\leq r(\Bsel) + \preindex(\Bsel)
			 &\text{(defn of index~\refeq{eq:def-index})} \\
	&= \nu_t(\Bsel) + 2\, r(\Bsel) + \rad_t(\Bsel)
			 &\text{(defns of pre-index~\refeq{eq:def-preindex})} \\
	&\leq \mu(\Bsel) + 3\,r(\Bsel) + 2\, \rad_t(\Bsel)
		&\text{(``clean run")} \\
	&\leq \mu(x_t,y_t) + 5\,r(\Bsel)
\end{align*}
For the last inequality, we use the fact that
    $\rad_t(\Bsel)\leq r(\Bsel)$
whenever $\Bsel$ is a parent ball.
\end{proof}

Now we are ready for the final regret computation. For a given $r=2^{-i}$, $i\in\N$, let $\F_r$ be the collection of all balls of radius $r$ that have been activated throughout the execution of the algorithm.
%A ball $B\in \F_r$ is called \emph{full} in round $t$ if $\rad_t(B)\leq r(B)$.
Note that in each round, if a parent ball is selected then some other ball is activated.  Thus, we can partition the rounds among active balls as follows: for each ball $B\in \F_r$, let $S_B$ be the set of rounds which consists of the round when $B$ was activated and all rounds $t$ when $B$ was selected and was not a parent ball.%
\footnote{A given ball $B$ can be selected even after it becomes a parent ball, but in such round some other ball $B’$ is activated, so this round is included in $S_{B'}$. }
It is easy to see that
	$|S_B| \leq O(r^{-2}\, \log T)$.
Moreover, by Lemma~\ref{lm:ZoomAlg-crux} and Corollary~\ref{cor:ZoomAlg-crux}
we have
	$\Delta(x_t,y_t) \leq 15\,r$
in each round $t\in S_B$.

If ball $B\in \F_r$ is activated in round $t$, then Corollary~\ref{cor:ZoomAlg-crux} asserts that its center $(x_t,y_t)$ lies in the set $\mPzoom$, as defined in~\refeq{eq:P-mu-r}. By the separation invariant, the centers of balls in $\F_r$  are within distance at least $r$ from one another. It follows that
    $|\F_r| \leq N_r$,
where $N_r$ is the $r$-zooming number.

Fixing some $r_0\in (0,1)$, note that in each rounds $t$ when a ball of radius $<r_0$ was selected, regret is $\Delta(x_t,y_t)\leq O(r_0)$, so the total regret from all such rounds is at most $O(r_0\, T)$. Therefore, contextual regret can be written as follows:
\begin{align*}
R(T)
	&= \textstyle{\sum_{t=1}^T}\, \Delta(x_t, y_t) \\
	&= O(r_0\, T) +
        \textstyle{
            \sum_{r=2^{-i}:\; r_0\leq r\leq 1}\;
            \sum_{B\in\F_r} \sum_{t\in S_B}
        }\;
			\Delta(x_t, y_t) \\
	&\leq O(r_0\, T) +
		\textstyle{
            \sum_{r=2^{-i}:\; r_0\leq r\leq 1}\;
            \sum_{B\in\F_r}
        }\;
        |S_B|\, O(r) \\
	&\leq O\left(
		r_0 T \,+\,
        \textstyle{
            \sum_{r=2^{-i}:\; r_0\leq r\leq 1}
        }\;
        \tfrac{1}{r}\, N_r \log (T)
	\right).
\end{align*}
The $N_r$-type regret guarantee in Theorem~\ref{thm:ZoomAlg} follows by taking $\inf$ on all $r_0\in (0,1)$.

%%%%%%%%%%%%%%%%%%%%%%%%%%%
\subsection{Improved regret bounds}
\label{subsec:stochastic-improved}

Let us provide regret bounds that take into account ``benign" context arrivals. The main difficulty here is to develop the corresponding definitions; the analysis then carries over without much modification. The added value is two-fold: first, we establish the intuition that benign context arrivals matter, and then the specific regret bound is used in Section~\ref{subsec:evolving} to match the result in~\cite{DynamicMAB-colt08}.

A crucial step in the proof of Theorem~\ref{thm:ZoomAlg} is to bound the number of active radius-$r$ balls by $N_r(\mPzoom)$, which is accomplished by observing that their centers form an $r$-packing $S$ of $\mPzoom$. We make this step more efficient, as follows. An active radius-$r$ ball is called \emph{full} if $\rad_t(B)\leq r$ for some round $t$. Note that each active ball is either full or a child of some other ball that is full. The number of children of a given ball is bounded by the doubling constant of the similarity space. Thus, it suffices to consider the number of active radius-$r$ balls that are full, which is at most $N_r(\mPzoom)$, and potentially much smaller.

Consider active radius-$r$ active balls that are full. Their centers form an $r$-packing $S$ of $\mPzoom$ with an additional property: each point $p\in S$ is assigned at least $1/r^2$ context arrivals $x_t$
so that
    $(x_t,y)\in B(p,r)$ for some arm $y$,
and each context arrival is assigned to at most one point in $S$.%
\footnote{Namely, each point $p\in S$ is assigned all contexts $x_t$ such that the corresponding ball is chosen in round $t$.}
A set $S\subset \mP$ with this property is called \emph{$r$-consistent} (with context arrivals). The \emph{adjusted $r$-packing number} of a set $\mP'\subset \mP$, denoted $\Nadj_r(\mP')$, is the maximal size of an $r$-consistent $r$-packing of $\mP'$.  It can be much smaller than the $r$-packing number of $\mP'$ if most context arrivals fall into a small region of the similarity space.

We make one further optimization, tailored to the application in Section~\ref{subsec:evolving}. Informally, we take advantage of context arrivals $x_t$ such that expected payoff $\mu(x_t,y)$ is either optimal or very suboptimal. A point $(x,y)\in \mP$ is called an \emph{$r$-winner} if for each
    $(x',y')\in B((x,y),\, 2r)$
it holds that
    $\mu(x',y') = \mu^*(x') $.
Let $\mWzoom$ be the set of all $r$-winners. It is easy to see that if $B$ is a radius-$r$ ball centered at an $r$-winner, and $B$ or its child is selected in a given round, then this round does not contribute to contextual regret. Therefore, it suffices to consider ($r$-consistent) $r$-packings of
    $\mPzoom\setminus \mWzoom$.

Our final guarantee is in terms of
    $\Nadj(\mPzoom\setminus \mWzoom)$,
which we term the \emph{adjusted $r$-zooming number}.

\begin{theorem}\label{thm:ZoomAlg-adjusted}
Consider the \problem{} with stochastic payoffs. The contextual regret $R(T)$ of the \ZoomAlg{} satisfies~\refeq{eq:regret-covNum}, where
$N_r$ is the adjusted $r$-zooming number and $C_0$ is the doubling constant of the similarity space times some absolute constant. Consequently, $R(T)$ satisfies the dimension-type guarantee~\refeq{eq:regret-dim}, where $d_c$ is the corresponding covering-type dimension.
\end{theorem}

\OMIT{ %%%%%%%
\begin{sketch}
It is easy to modify the proof of Theorem~\ref{thm:ZoomAlg} so that instead of counting the balls of radius $r$ that are activated by the algorithm, we count the balls of radius $r$ that become full. There can be at most $N_r$ such balls. Now, each active ball is either full at time $T$, or a child of some ball that has become full. The theorem follows since each parent ball has at most $\Cdbl$ children.
\end{sketch}
} %%%%%%%%

%%%%%%%%%%%%%%%%%%%%%%
%%%%%%%%%%%%%%%%%%%%%%%
\section{Lower bounds}
\label{sec:LBs}

\newcommand{\regretUB}{R^{\text{UB}}}

We match the upper bound in Theorem~\ref{thm:ZoomAlg} up to $O(\log T)$ factors. Our lower bound is very general: it applies to an arbitrary product similarity space, and moreover for a given similarity space it matches, up to $O(\log T)$ factors, any fixed value of the upper bound (as explained below).

We construct a distribution $\mI$ over problem instances on a given metric space, so that the lower bound is for a problem instance drawn from this distribution. A single problem instance would not suffice to establish a lower bound because a trivial algorithm that picks arm $y^*(x)$ for each context $x$ will achieve regret $0$.

The distribution $\mI$ satisfies the following two properties: the upper bound in Theorem~\ref{thm:ZoomAlg} is uniformly bounded from above by some number $R$, and any algorithm must incur regret at least $\Omega(R/\log T)$ in expectation over $\mI$. Moreover, we constrict such $\mI$ for every possible value of the upper bound in Theorem~\ref{thm:ZoomAlg} on a given metric space, i.e. not just for problem instances that are ``hard" for this metric space.

To formulate our result, let $\regretUB_\mu(T)$ denote the upper bound in Theorem~\ref{thm:ZoomAlg}, i.e. is the right-hand side of~\refeq{eq:regret-covNum} where $N_r = N_r(\mP_{\mu,r})$ is the $r$-zooming number. Let $\regretUB(T)$ denote the pessimistic version of this bound, namely right-hand side of~\refeq{eq:regret-covNum} where $N_r = N_r(\mP)$ is the packing number of $\mP$.

\begin{theorem}\label{thm:LBs-strong}
Consider the \problem\ with stochastic payoffs,  Let $(\mP,\D)$ be a product similarity space. Fix an arbitrary time horizon $T$ and a positive number
    $R\leq \regretUB(T)$.
Then there exists a distribution $\mI$ over problem instances on $(\mP,\D)$ with the following two properties:
\begin{OneLiners}
\item[(a)] $\regretUB_{\mu}(T) \leq O(R)$
for each problem instance in $\mathtt{support}(\mI)$.

\item[(b)] for any contextual bandit algorithm it holds that
    $\E_{\mI}[R(T)] \geq \Omega(R/\log T)$,

\end{OneLiners}
\end{theorem}

To prove this theorem, we build on the lower-bounding technique from~\cite{bandits-exp3}, and its extension to (context-free) bandits in metric spaces in~\cite{Bobby-nips04}. In particular, we use the basic \emph{needle-in-the-haystack} example from~\cite{bandits-exp3}, where the ``haystack" consists of several arms with expected payoff $\tfrac12$, and the ``needle" is an arm whose expected payoff is slightly higher.

\xhdr{The lower-bounding construction.}
Our construction is parameterized by two numbers: $r\in (0,\tfrac12]$ and
    $N\leq N_r(\mP)$,
where $N_r(\mP)$ is the $r$-packing number of $\mP$. Given these parameters, we construct a collection $\mI = \mI_{N,r}$ of $\Theta(N)$ problem instances as follows.

Let $N_{\TX,r}$ be the $r$-packing number of $X$ in the context space, and let $N_{\TY,r}$ be the $r$-packing number of $Y$ in the arms space. Note that
    $N_r(\mP) = N_{\TX,r} \times N_{\TY,r} $.
For simplicity, let us assume that
    $N = n_\TX \, n_\TY$,
where
	$1\leq n_X\leq N_{\TX,r}$ and $2\leq n_\TY \leq N_{\TY,r}$.

An \emph{$r$-net} is the set $S$ of points in a metric space such that any two points in $S$ are at distance $> r$ from each other, and each point in the metric space is within distance $\leq r$ from some point in $S$. Recall that any $r$-net  on the context space has size at least $N_{\TX,r}$. Let $S_\TX$ be an arbitrary set of $n_\TX$ points from one such $r$-net. Similarly, let $S_\TY$ be an arbitrary set of $n_\TY$ points from some $r$-net on the arms space. The sequence $\arr$ of context arrivals is any fixed permutation over the points in $S_\TX$, repeated indefinitely.

All problem instances in $\mI$ have 0-1 payoffs. For each $x\in S_\TX$ we construct a needle-in-the-haystack example on the set $S_\TY$. Namely, we pick one point $y^*(x)\in S_\TY$ to be the ``needle", and define
	$\mu(x,y^*(x)) = \tfrac{1}{2}+ \tfrac{r}{4}$,
and
	$\mu(x,y) = \tfrac12+\tfrac{r}{8}$
for each $y\in S_\TY\setminus \{y^*(x)\}$. We smoothen the expected payoffs so that far from $S_\TX\times S_\TY$ expected payoffs are $\tfrac12$ and the Lipschitz condition~\refeq{eq:LipschitzD} holds:
\begin{align}\label{eq:LB-construction}
\mu(x,y)
	\triangleq
	\max_{(x_0,\,y_0)\in S_\TX\times S_\TY}\;
	\max\left(\,
		\tfrac12,\; \mu(x_0,y_0) - \D_\TX(x,x_0) - \D_\TY(y,y_0)\
	\,\right).
\end{align}
Note that we obtain a distinct problem instance for each function
    $y^*(\cdot): S_\TX\to S_\TY$.
This completes our construction.

\xhdr{Analysis.} The useful properties of the above construction are summarized in the following lemma:

\begin{lemma}\label{lm:LB}
Fix $r\in (0,\tfrac12]$ and $N\leq N_r(\mP)$. Let $\mI = \mI_{N,r}$ and $T_0 = N\,r^{-2}$. Then:
\begin{OneLiners}
\item[(i)] for each problem instance in $\mI$ it holds that
    $\regretUB_{\mu}(T_0) \leq O(N/r)(\log T_0)$.

\item[(ii)] any contextual bandit algorithm has regret
	$\E_{\mI}[R(T_0)] \geq \Omega(N/r)$
for a problem instance chosen uniformly at random from $\mI$.
\end{OneLiners}
\end{lemma}

For the lower bound in Lemma~\ref{lm:LB}, the idea is that in $T$ rounds each context in $S_\TX$ contributes $\Omega(|S_\TY|/r)$ to contextual regret, resulting in total contextual regret $\Omega(N/r)$.

Before we proceed to prove Lemma~\ref{lm:LB}, let us use it to derive Theorem~\ref{thm:LBs-strong}.
Fix an arbitrary time horizon $T$ and a positive number
    $R\leq \regretUB(T)$.
Recall that since $N_r(\mP)$ is  non-increasing in $r$, for some constant $C>0$ it holds that
\begin{align}\label{eq:regret-covNum-inLB}
    \regretUB(T) = C\times \textstyle{\inf_{r\in (0,1)}}\,
    		\left( r T	+ \tfrac{1}{r}\,N_r(\mP) \log T \right).
\end{align}

\begin{claim}
Let
   $ r = \frac{R}{2C\,T(1+\log T)}$.
Then  $r \leq \tfrac12$ and $T r^2 \leq N_r(\mP)$.
\end{claim}

\begin{proof}
Denote $k(r) = N_r(\mP)$ and consider function $f(r) \triangleq k(r)/r^2$. This function is non-increasing in $r$; $f(1) = 1$ and $f(r)\to \infty$ for $r\to 0$. Therefore there exists $r_0\in (0,1)$ such that
    $f(r_0) \leq T \leq f(r_0/2)$.
Re-writing this, we obtain
\begin{align*}
k(r_0) \leq T\, r_0^2 \leq 4\, k(r_0/2).
\end{align*}
It follows that
\begin{align*}
R &\leq \regretUB(T) \leq C( T r_0 + \tfrac{1}{r_0}\, k(r_0) \log T)
    \leq C\, T r_0(1+ \log T).
\end{align*}
Thus
    $r \leq r_0/2$
and finally
    $T\, r^2 \leq T\, r_0^2/4 \leq k(r_0/2) \leq k(r) = N_r(\mP)$.
\end{proof}

So, Lemma~\ref{lm:LB} with
   $ r \triangleq \frac{R}{2C\, T(1+\log T)}$
and
    $N \triangleq T\, r^2$.
implies Theorem~\ref{thm:LBs-strong}.

%%%%%%%%%%%%%%%%%%%%%%%%%%%%%%%%
\subsection{Proof of Lemma~\ref{lm:LB}}

\begin{claim}
Collection $\mI$ consists of valid instances of contextual MAB problem with similarity space $(\mP,\D)$.
\end{claim}
\begin{proof}
We need to prove that each problem instance in $\mP$ satisfies the Lipschitz condition~\refeq{eq:LipschitzD}. Assume the Lipschitz condition~\refeq{eq:LipschitzD} is violated for some points
    $(x,y),\, (x',y') \in X\times Y$.
For brevity, let
    $p = (x,y)$, $p' = (x',y')$,
and let us write $\mu(p) \triangleq \mu(x,y)$. Then
    $|\mu(p)-\mu(p')| >\D(p,p')$.

By~\refeq{eq:LB-construction},
    $\mu(\cdot)\in [\tfrac12, \tfrac12+\tfrac{r}{4}]$,
so
    $\D(p,p')< \tfrac{r}{4}$.

Without loss of generality, $\mu(p)>\mu(p')$. In particular, $\mu(p)>\tfrac12$. Therefore there exists
    $p_0 = (x_0,y_0) \in S_\TX \times S_\TY$
such that
    $\D(p,p_0) < \tfrac{r}{4}$.
Then
    $\D(p',p_0) < \tfrac{r}{2}$
by triangle inequality.

Now, for any other
    $p'_0 \in S_\TX\times S_\TY$
it holds that
    $\D(p_0, p'_0)> r$,
and thus by triangle inequality
    $\D(p, p'_0)> \tfrac{3r}{4}$
and
    $\D(p', p'_0)> \tfrac{r}{2}$.
It follows that~\refeq{eq:LB-construction} can be simplified as follows:
\begin{align*}
\begin{cases}
\mu(p) &= \max(\tfrac12,\, \mu(p_0) - \D(p,p_0)), \\
\mu(p')&= \max(\tfrac12,\, \mu(p_0) - \D(p',p_0)).
\end{cases}
\end{align*}
Therefore
\begin{align*}
|\mu(p) - \mu(p')|
    &=  \mu(p) - \mu(p') \\
    &= (\mu(p_0) - \D(p,p_0))
        - \max(\tfrac12,\, \mu(p_0) - \D(p',p_0)) \\
    &\leq (\mu(p_0) - \D(p,p_0)) - (\mu(p_0) - \D(p',p_0))) \\
    &= \D(p',p_0) - \D(p,p_0)
    \leq \D(p,p').
\end{align*}
So we have obtained a contradiction.
\end{proof}

\begin{claim}
For each instance in $\mP$ and $T_0 = N\, r^{-2}$ it holds that
    $\regretUB_{\mu}(T_0) \leq O(N/r)(\log T_0)$.
\end{claim}
\begin{proof}
Recall that $\regretUB_\mu(T_0)$ is the right-hand side of~\refeq{eq:regret-covNum} with $N_r = N_r(\mP_{\mu,r})$, where $\mP_{\mu,r}$ is defined by~\refeq{eq:P-mu-r}.

Fix $r'>0$. It is easy to see that
\begin{align*}
\mP_{\mu,\,r'} \subset \cup_{p\in S_\TX\times S_\TY}\; B(p,\tfrac{r}{4}).
\end{align*}
It follows that $N_{r'}(\mP_{\mu,r'}) \leq N$ whenever $r'\geq \tfrac{r}{4}$. Therefore, taking $r_0 = \tfrac{r}{4}$ in~\refeq{eq:regret-covNum}, we obtain
\begin{align*}
\regretUB_{\mu}(T_0)
    \leq O(r T_0 + \tfrac{N}{r} \log T_0)
    = O(N/r)(\log T_0).
\end{align*}
\end{proof}

\begin{claim}
Fix a contextual bandit algorithm \A. This algorithm has regret
	$\E_{\mI}[R(T_0)] \geq \Omega(N/r)$
for a problem instance chosen uniformly at random from $\mI$, where $T_0 = N\, r^{-2}$.
\end{claim}

\begin{proof}
Let $R(x,T)$ be the contribution of each context $x\in S_\TX$ to contextual regret:
\begin{align*}
R(x,T) = \sum_{t:\, x_t=x}\; \mu^*(x) - \mu(x,y_t),
\end{align*}
where $y_t$ is the arm chosen by the algorithm in round $t$. Our goal is to show that
    $R(x,T_0)\geq \Omega(r\,n_\TY)$.

We will consider each context $x\in S_\TX$ separately: the rounds when $x$ arrives form an instance $I_x$ of a context-free bandit problem that lasts for
    $T_0/n_\TX = n_\TY\, r^{-2}$
rounds, where expected payoffs are  given by $\mu(x,\cdot)$ as defined in~\refeq{eq:LB-construction}. Let $\mI_x$ be the family of all such instances $I_x$.

A uniform distribution over $\mI$ can be reformulated as follows: for each $x\in S_\TX$, pick the ``needle" $y^*(x)$ independently and uniformly at random from $S_\TY$. This induces a uniform distribution over instances in $\mI_x$, for each context $x\in S_\TX$. Informally, knowing full or partial information about $y^*(x)$ for some $x$ reveals no information whatsoever about $y^*(x')$ for any $x'\neq x$.

Formally, the contextual bandit algorithm $\A$ induces a bandit algorithm $\A_x$ for $I_x$, for each context $x\in S_\TX$: the $\A_x$ simulates the problem instance for $\A$ for all contexts $x'\neq x$ (starting from the ``needles" $y^*(x')$ chosen independently and uniformly at random from $S_\TY$). Then $\A_x$ has expected regret $R_x(T)$ which satisfies
    $\E[\, R(T)\,] =  \E[\, R(x,T) \,]$,
where the expectations on both sides are over the randomness in the respective algorithm and the random choice of the problem instance (resp., from $\mI_x$ and from $\mI$).

Thus, it remains to handle each $\mI_x$ separately: i.e., to prove that the expected regret of any bandit algorithm on an instance drawn uniformly at random from $\mI_x$ is at least $\Omega(r\,n_\TY)$. We use the KL-divergence technique that originated in~\cite{bandits-exp3}. If the set of arms were exactly $S_\TY$, then the desired lower bound would follow from~\cite{bandits-exp3} directly. To handle the problem instances in $\mI_x$, we use an extension of the technique from~\cite{bandits-exp3}, which is implicit in~\cite{Bobby-nips04} and encapsulated as a stand-alone theorem in \cite{LipschitzMAB-merged-arxiv}. We restate this theorem as  Theorem~\ref{thm:KL-div} in Appendix~\ref{app:KL-div}.

It is easy to check that the family $\mI_x$ of problem instances satisfies the preconditions in Theorem~\ref{thm:KL-div}. Fix $x\in S_\TX$. For a given choice of the ``needle" $y^*=y^*(x)\in S_\TY$, let
    $\mu(x,y\,|\,y^*)$.
be the expected payoff of each arm $y$, and let
    $\nu_{y^*}(\cdot) = \mu(x,\cdot \,|\, y^*)$
be the corresponding payoff function for the bandit instance $I_x$. Then
    $\{ \nu_{y^*} \}$, $y^*\in S_\TY$
is an ``$(\eps,k)$-ensemble'' for $\eps = \tfrac{r}{8}$ and $k=|S_\TY|$.
\end{proof}

%%%%%%%%%%%%%%%%%%%%%%
%%%%%%%%%%%%%%%%%%%%%%%

\section{Applications of contextual zooming}
\label{sec:apps}

We describe several applications of contextual zooming: to MAB with slow adversarial change (Section~\ref{subsec:driftProblem}), to MAB with stochastically evolving payoffs (Section~\ref{subsec:evolving}), and to the ``sleeping bandits" problem (Section~\ref{subsec:sleepingMAB}). In particular, we recover some of the main results in~\cite{DynamicMAB-colt08} and~\cite{sleeping-colt08}. Also, in Section~\ref{subsec:ZoomingRBA} we discuss a recent application of contextual zooming to bandit learning-to-rank, which has been published in~\cite{ZoomingRBA-icml10}.

\subsection{MAB with slow adversarial change}
\label{subsec:driftProblem}

% Consider a  adversarial MAB problem in which an adversary
% is constrained to change the expected payoffs of each arm

Consider the (context-free) adversarial MAB problem in which expected payoffs of each arm change over time \emph{gradually}. Specifically, we assume that expected payoff of each arm $y$ changes by at most $\sigma_y$ in each round, for some a-priori known \emph{volatilities} $\sigma_y$. The algorithm's goal here is continuously adapt to the changing environment, rather than converge to the best fixed mapping from contexts to arms. We call this setting the {\bf\em \driftProblem}.

Formally, our benchmark is a fictitious  algorithm which in each round selects an arm that maximizes expected payoff for the current context. The difference in expected payoff between this benchmark and a given algorithm is called \emph{dynamic regret} of this algorithm. It is easy to see that the worst-case dynamic regret of any algorithm cannot be sublinear in time.%
\footnote{For example, consider problem instances with two arms such that the payoff of each arm in each round is either $\tfrac12$ or $\tfrac12+\sigma$ (and can change from round to round). Over this family of problem instances, dynamic regret in $T$ rounds is at least $\tfrac12\,\sigma T$.}
We are primarily interested in algorithm's long-term performance, as quantified by \emph{average} dynamic regret
    $\hat{R}(T) \triangleq R(T)/T$.
Our goal is to bound the limit
    $\lim_{T\to \infty} \hat{R}(T)$
in terms of the parameters: the number of arms and the volatilities $\sigma_y$. (In general, such upper bound is non-trivial as long as it is smaller than 1, since all payoffs are at most 1.)

We restate this setting as a contextual MAB problem with stochastic payoffs in which the $t$-th context arrival is simply $x_t=t$. Then $\mu(t,y)$ is the expected payoff of arm $y$ at time $t$, and dynamic regret coincides with contextual regret specialized to the case $x_t=t$. Each arm $y$ satisfies a ``temporal constraint":
\begin{align}\label{eq:driftingMAB-Lip}
	|\mu(t,y)-\mu(t',y)| \leq \sigma_y\,|t-t'|
\end{align}
for some constant $\sigma_y$. To set up the corresponding similarity space $(\mP,\D)$, let
    $\mP = [T]\times Y$,
and
\begin{align}\label{eq:driftingMAB-D}
    \D((t,y),\, (t',y'))
        = \min(1,\;  \sigma_y\,|t-t'| + \indicator{y\neq y'}).
\end{align}

\OMIT{ %%%%%%%
The similarity distance is simply
\begin{align}\label{eq:driftingMAB-similarity}
\D((t,y),\, (t',y'))
    = \begin{cases}
        \D_y(t,t')  & \text{if $y = y'$}, \\
        1           & \text{otherwise}.
    \end{cases}
\end{align}
} %%%%%%%

Our solution for the \driftProblem{} is the \ZoomAlg{} parameterized by the similarity space $(\mP,\D)$. To obtain guarantees for the long-term performance, we run contextual zooming with a suitably chosen time horizon $T_0$, and restart it every $T_0$ rounds; we call this version \emph{contextual zooming with period $T_0$}. Periodically restarting the algorithm is a simple way to prevent the change over time from becoming too large; it suffices to obtain strong provable guarantees.

The general provable guarantees are provided by Theorem~\ref{thm:ZoomAlg} and Theorem~\ref{thm:ZoomAlg-adjusted}. Below we work out some specific, tractable corollaries.

\begin{corollary}\label{cor:driftingMAB-pessimistic}
Consider the \driftProblem{} with $k$ arms and volatilities $\sigma_y\equiv \sigma$.
Contextual zooming with period $T_0$ has average dynamic regret
    $\hat{R}(T) = O(k\sigma \log T_0)^{1/3}$,
whenever
    $T\geq T_0\geq (\tfrac{k}{\sigma^2})^{1/3}\, \log \tfrac{k}{\sigma} $.
\end{corollary}

\begin{proof}
It suffices to upper-bound regret in a single period. Indeed, if $R(T_0) \leq R$ for any problem instance, then $R(T) \leq R\, \cel{T/T_0}$ for any $T>T_0$. It follows that
    $\hat{R}(T) \leq 2\,\hat{R}(T_0)$.
Therefore, from here on we can focus on analyzing contextual zooming itself, rather than contextual zooming with a period.

The main step is to derive the regret bound~\refeq{eq:regret-covNum} with a specific upper bound on $N_r$. We will show that
\begin{align}\label{eq:driftingMAB-pessimistic}
    \text{dynamic regret $R(\cdot)$ satisfies~\refeq{eq:regret-covNum} with }
    N_r \leq k\,\cel{\tfrac{T\sigma}{r}}.
\end{align}
Plugging $N_r \leq k\,(1+\tfrac{T\sigma}{r})$ into~\refeq{eq:regret-covNum} and taking $r_0 = (k\sigma \log T)^{1/3}$ we obtain\footnote{This choice of $r_0$ minimizes the $\inf$ expression in~\refeq{eq:regret-covNum} up to constant factors by equating the two summands.}
\begin{align*}
    R(T) \leq O(T) (k\sigma \log T)^{1/3} + O(\tfrac{k^2}{\sigma})^{1/3}(\log T)
    \qquad \forall T\geq 1.
\end{align*}
Therefore, for any
        $T\geq (\tfrac{k}{\sigma^2})^{1/3}\, \log \tfrac{k}{\sigma} $
we have
    $\hat{R}(T) = O(k\sigma \log T)^{1/3}$.

It remains to prove~\refeq{eq:driftingMAB-pessimistic}. We use a pessimistic version of Theorem~\ref{thm:ZoomAlg}:~\refeq{eq:regret-covNum} with $N_r = N_r(\mP)$, the  $r$-packing number of $\mP$. Fix
    $r\in (0,1]$.
For any $r$-packing $S$ of $\mP$ and each arm $y$, each time interval $I$ of duration
    $\Delta_r\triangleq r/\sigma$
provides at most one point for $S$: there exists at most one time $t\in I$ such that $(t,y)\in S$. Since there are at most $\cel{T/\Delta_r}$ such intervals $I$, it follows that
    $N_r(\mP) \leq k\, \cel{T/\Delta_r} \leq k\,(1+T \tfrac{\sigma}{r})$.
\end{proof}

The restriction $\sigma_y \equiv \sigma$ is non-essential: it is not hard to obtain the same bound with
    $\sigma = \tfrac{1}{k} \sum_y \sigma_y$.
Modifying the construction in Section~\ref{sec:LBs} (details omitted from this version) one can show that
Corollary~\ref{cor:driftingMAB-pessimistic} is optimal up to $O(\log T)$ factors.
%, see Section~\ref{app:driftingMAB}.

\xhdr{Drifting MAB with spatial constraints.}
The temporal version ($x_t=t$) of our contextual MAB setting with stochastic payoffs subsumes the drifting MAB problem and furthermore allows to combine the temporal constraints~\refeq{eq:driftingMAB-Lip} described above (for each arm, across time) with ``spatial constraints" (for each time, across arms). To the best of our knowledge, such MAB models are quite rare in the literature.\footnote{The only other MAB model with this flavor that we are aware of, found in~\cite{Hazan-soda09}, combines linear payoffs and bounded ``total variation" (aggregate temporal change) of the cost functions.}
A clean example is
\begin{align}\label{eq:driftingMAB-spatial}
    \D((t,y),\, (t',y')) = \min(1,\; \sigma\,|t-t'| + \D_\TY(y,y') ),
\end{align}
where $(Y,\D_\TY)$ is the arms space. For this example, we can obtain an analog of Corollary~\ref{cor:driftingMAB-pessimistic}, where the regret bound depends on the covering dimension of the arms space $(Y,\D_\TY)$.

\begin{corollary}\label{cor:driftingMAB-spatial}
Consider the \driftProblem{} with spatial constraints~\refeq{eq:driftingMAB-spatial}, where $\sigma$ is the volatility. Let $d$ be the covering dimension of the arms space, with multiplier $k$. Contextual zooming with period $T_0$ has average dynamic regret
    $\hat{R}(T) = O(k\, \sigma \log T_0)^{\tfrac{1}{d+3}}$,
whenever
    $T\geq T_0
        \geq  k^{\tfrac{1}{d+3}} \; \sigma^{-\tfrac{d+2}{d+3}} \;
            \log \tfrac{k}{\sigma} $.
\end{corollary}

\begin{note}{Remark.}
We obtain Corollary~\ref{cor:driftingMAB-pessimistic} as a special case by setting $d=0$.
\end{note}

\begin{proof}
It suffices to bound $\hat{R}(T_0)$ for (non-periodic) contextual zooming. First we bound the $r$-covering number of the similarity space $(\mP,\D)$:
\begin{align*}
N_r(\mP)
    = N_r^\TX(X) \times N_r^\TY(Y)
    \leq \cel{\tfrac{T\sigma}{r}}\, k\,r^{-d},
\end{align*}
where $N_r^\TX(\cdot)$ is the $r$-covering number in the context space, and $N_r^\TY(\cdot)$ is that in the arms space. We worked out the former for  Corollary~\ref{cor:driftingMAB-pessimistic}. Plugging this into~\refeq{eq:regret-covNum} and taking
    $r_0 = (k\, \sigma \log T)^{1/(3+d)}$,
we obtain
\begin{align*}
    R(T) \leq O(T) (k\sigma \log T)^{\tfrac{1}{d+3}}
        + O\left(
                k^{\tfrac{2}{d+3}}\,
                \sigma^{\tfrac{d+1}{d+3}} \,
                \log T
        \right)
    \qquad \forall T\geq 1.
\end{align*}
The desired bound on $\hat{R}(T_0)$ follows easily.
\end{proof}

\OMIT{ %%%%%%
Moreover, the temporal constraints~\refeq{eq:driftingMAB-Lip} can take an arbitrary metric $D_y$ on contexts for each arm $y\in Y$. One way to generalize the drifting MAB setting is to assume, for some fixed $\gamma\in (0,1]$, that
\begin{align}\label{eq:driftingMAB-gamma}
	\D_y(t,t') = \min(1,\; \sigma_y\, |t-t'|^\gamma).
\end{align}
E.g., for $\gamma=\tfrac12$ this corresponds to a typical behavior of a random walk.\footnote{A $t$-step displacement of a random walk with step $\pm \sigma$ is $\Theta(\sigma \sqrt{t})$ with high probability.}
} %%%%%

%%%%%%%%%%
\subsection{Bandits with stochastically evolving payoffs}
\label{subsec:evolving}
% in~\cite{DynamicMAB-colt08} under the name \emph{Dynamic MAB problem}.

We consider a special case of \driftProblem{} in which expected payoffs of each arm evolve over time according to a stochastic process with a uniform stationary distribution. We obtain improved regret bounds for contextual zooming, taking advantage of the full power of our analysis in Section~\ref{sec:ZoomAlg}.

In particular, we address a version in which the stochastic process is a random walk with step $\pm \sigma$. This version has been previously studied in~\cite{DynamicMAB-colt08} under the name ``Dynamic MAB". For the main case ($\sigma_i\equiv \sigma$), our regret bound for Dynamic MAB matches that in~\cite{DynamicMAB-colt08}.

To improve the flow of the paper, the proofs are deferred to  Appendix~\ref{app:evolving}.

\xhdr{Uniform marginals.} First we address the general version that we call {\em drifting MAB with uniform marginals}. Formally, we assume that expected payoffs $\mu(\cdot,y)$ of each arm $y$ evolve over time according to some stochastic process $\Gamma_y$ that satisfies~\refeq{eq:driftingMAB-Lip}. We assume that the processes $\Gamma_y$, $y\in Y$ are mutually independent, and moreover that the marginal distributions $\mu(t,y)$ are uniform on $[0,1]$, for each time $t$ and each arm $y$.~\footnote{E.g. this assumption is satisfied by any Markov Chain on $[0,1]$ with stationary initial distribution. }  We are interested in $\E_\Gamma[\hat{R}(T)]$,
average dynamic regret in expectation over the processes $\Gamma_y$.

We obtain a stronger version of~\refeq{eq:driftingMAB-pessimistic} via Theorem~\ref{thm:ZoomAlg-adjusted}. To use this theorem, we need to bound the adjusted $r$-zooming number, call it $N_r$. We show that
\begin{align}\label{eq:driftingMAB-uniform}
\E_{\Gamma}[N_r] = O(kr)\cel{\tfrac{T\sigma}{r}}
    \text{ and }
\left( r < \sigma^{1/3} \Rightarrow N_r=0 \right).
\end{align}
\noindent Then we obtain a different bound on dynamic regret, which is stronger than Corollary~\ref{cor:driftingMAB-pessimistic} for $k<\sigma^{-1/2}$.

\begin{corollary}\label{cor:driftingMAB-marginals}
Consider drifting MAB with uniform marginals, with $k$ arms and volatilities $\sigma_y\equiv \sigma$. Contextual zooming with period $T_0$ satisfies
    $\E_\Gamma[\hat{R}(T)] = O(k\,\sigma^{2/3}\,\log T_0)$,
whenever
    $T\geq T_0\geq \sigma^{-2/3} \log\tfrac{1}{\sigma}$.
\end{corollary}

The crux of the proof is to show~\refeq{eq:driftingMAB-uniform}. Interestingly, it involves using all three optimizations in Theorem~\ref{thm:ZoomAlg-adjusted}:
    $N_r(\mPzoom)$,
    $N_r(\mPzoom\setminus\mWzoom)$ and
    $\Nadj_r(\cdot)$,
whereas any two of them do not seem to suffice. The rest is a straightforward computation similar to the one in Corollary~\ref{cor:driftingMAB-pessimistic}.
%See the full version for details.
%See Appendix~\ref{app:driftingMAB} for details.

\xhdr{Dynamic MAB.}
Let us consider the Dynamic MAB problem from~\cite{DynamicMAB-colt08}. Here for each arm $y$ the stochastic process $\Gamma_y$ is a random walk with step $\pm \sigma_y$. To ensure that the random walk stays within the interval $[0,1]$, we assume reflecting boundaries. Formally, we assume that $1/\sigma_y\in \N$, and once a boundary is reached, the next step is deterministically in the opposite direction.\footnote{\cite{DynamicMAB-colt08} has a slightly more general setup which does not require $1/\sigma_y\in \N$.}

According to a well-known fact about random walks,%
\footnote{For example, this follows as a simple application of Azuma-Hoeffding inequality.}
\begin{align}\label{eq:DynamicMAB-Lip}
\Pr\left[
	|\mu(t,y)-\mu(t',y)| \leq O(\sigma_y\,|t-t'|^{1/2} \log T_0)
\right] \geq 1-T_0^{-3}
\quad \text{if $|t-t'|\leq T_0$}.
\end{align}
We use contextual zooming with period $T_0$, but we parameterize it by a different similarity space
    $(\mP,\D_{T_0})$
that we define according to~\refeq{eq:DynamicMAB-Lip}. Namely, we set
\begin{align}\label{eq:DynamicD}
    \D_{T_0}((t,y),\, (t',y'))
        = \min(1,\; \sigma_y\,|t-t'|^{1/2} \log T_0  + \indicator{y\neq y'}).
\end{align}

\noindent The following corollary is proved using the same technique as
Corollary~\ref{cor:driftingMAB-marginals}:
%see the full version for details.
%see Appendix~\ref{app:driftingMAB} for details.

\begin{corollary}\label{cor:DynamicMAB}
Consider the Dynamic MAB problem with $k$ arms and volatilities $\sigma_y\equiv \sigma$. Let $\mathtt{ALG}_{T_0}$ denote the contextual zooming algorotihm with period $T_0$ which is parameterized by the similarity space $(\mP,\D_{T_0})$. Then
    $\mathtt{ALG}_{T_0}$
satisfies
    $\E_\Gamma[\hat{R}(T)] = O(k\,\sigma\,\log^2 T_0)$,
whenever
    $T\geq T_0\geq \tfrac{1}{\sigma} \log\tfrac{1}{\sigma}$.
\end{corollary}

%%%%%%%%%%%%
\subsection{Sleeping bandits}
\label{subsec:sleepingMAB}

The \emph{sleeping bandits} problem~\cite{sleeping-colt08} is an extension of MAB where in each round some arms can be ``asleep", i.e. not available in this round. One of the main results in~\cite{sleeping-colt08} is on sleeping bandits with stochastic payoffs. We recover this result using contextual zooming.

We model sleeping bandits as \problem{} where each context arrival $x_t$ corresponds to the set of arms that are ``awake" in this round. More precisely, for every subset $S\subset Y$ of arms there is a distinct context $x_S$, and
    $\mP = \{(x_S,y):\, y\in S\subset Y  \}$.
is the set of feasible context-arm pairs. The similarity distance is simply
    $\D((x,y),\, (x',y')) = \indicator{y\neq y'}$.
Note that the Lipschitz condition~\refeq{eq:LipschitzD} is satisfied.

For this setting, contextual zooming essentially reduces to the ``highest awake index" algorithm in~\cite{sleeping-colt08}. In fact, we can re-derive the result~\cite{sleeping-colt08} on sleeping MAB with stochastic payoffs as an easy corollary of Theorem~\ref{thm:ZoomAlg}.

\begin{corollary}
Consider the sleeping MAB problem with stochastic payoffs. Order the arms so that their expected payoffs are
    $\mu_1\leq \mu_2 \leq \ldots \leq \mu_n$,
where $n$ is the number of arms. Let $\Delta_i = \mu_{i+1}-\mu_i$. Then
\begin{align*}
R(T) \leq \inf_{r>0} \left(
        rT + \sum_{i:\, \Delta_i>r}\; \frac{O(\log T)}{\Delta_i}
    \right).
\end{align*}
\end{corollary}

\begin{proof}
The $r$-zooming number $N_r(\mP_{\mu,r})$ is equal to the number of distinct \emph{arms} in $\mP_{\mu,r}$, i.e. the number of arms $i\in Y$ such that
    $\Delta(x,i)\leq 12r$
for some context $x$. Note that for a given arm $i$, the quantity $\Delta(x,i)$ is minimized when the set of awake arms is $S = \{i,i+1\}$. Therefore, $N_r(\mP_{\mu,r})$ is equal to the number of arms $i\in Y$ such that
    $\Delta_i\leq 12r$. It follows that
\begin{align*}
N_{r>r_0}(\mP_{\mu,r})
    &= \textstyle{\sum_{i=1}^n} \, \indicator{\Delta_i\leq 12r}. \\
\textstyle{\sum_{r>r_0}}\; \tfrac{1}{r} N_{r>r_0}(\mP_{\mu,r})
    &= \textstyle{\sum_{r>r_0}}\; \textstyle{\sum_{i=1}^n}\;
                \tfrac{1}{r} \, \indicator{\Delta_i\leq 12r} \\
    &=\textstyle{\sum_{i=1}^n}\; \textstyle{\sum_{r>r_0}}\;
                \tfrac{1}{r} \, \indicator{\Delta_i\leq 12r} \\
    &=\textstyle{\sum_{i:\, \Delta_i>r_0}}\; O(\tfrac{1}{\Delta_i}).\\
R(T)
    &\leq \inf_{r_0>0} \left(
        r_0\,T + O(\log T)\; \textstyle{\sum_{r>r_0}}\; \tfrac{1}{r} N_r(\mP_{\mu,r})
    \right) \\
    &\leq \inf_{r_0>0} \left(
        r_0\,T + O(\log T)\;
            \textstyle{\sum_{i:\, \Delta_i>r_0}}\; O(\tfrac{1}{\Delta_i})
    \right),
\end{align*}
as required. (In the above equations, $\sum_{r>r_0}$ denotes the sum over all $r=2^{-j}>r_0$ such that $j\in \N$.)
\end{proof}

Moreover, the \problem{} extends the sleeping bandits setting by incorporating similarity information on arms. The \ZoomAlg{} (and its analysis) applies, and is geared to exploit this additional similarity information.

%%%%%%%%%%%%%%%%%%%%%%%%%%%
\subsection{Bandit learning-to-rank}
\label{subsec:ZoomingRBA}

Following a preliminary publication of this paper on {\tt arxiv.org},
contextual zooming has been applied in~\cite{ZoomingRBA-icml10} to bandit learning-to-rank. Interestingly, the ``contexts" studied in~\cite{ZoomingRBA-icml10} are very different from what we considered so far.

The basic setting,  motivated by web search, was introduced in~\cite{RBA-icml08}. In each round a new user arrives. The algorithm selects a ranked list of $k$ documents and presents it to the user who clicks on at most one document, namely on the first document that (s)he finds relevant. A user is specified by a binary vector over documents. The goal is to minimize \emph{abandonment}: the number of rounds with no clicks.

\cite{ZoomingRBA-icml10} study an extension in which metric similarity information is available. They consider a version with \emph{stochastic payoffs}: in each round,
the user vector is an independent sample from a fixed distribution, and assume a Lipschitz-style condition that connects expected clicks with the metric space. They run a separate bandit algorithm (e.g., contextual zooming) for each of the $k$ ``slots" in the ranking. Without loss of generality, in each round the documents are selected sequentially, in the top-down order. Since a document in slot $i$ is clicked in a given round only if all higher ranked documents are not relevant, they treat the set of documents in the higher slots as a \emph{context} for the $i$-th algorithm. The Lipschitz-style condition on expected clicks suffices to guarantee the corresponding  Lipschitz-style condition on contexts.

%%%%%%%%%%%%%%%%%%%%%%%%%
%%%%%%%%%%%%%%%%%%%%%%%%%
%%%%%%%%%%%%%%%%%%%%%%%%%

\section{Bandits with stochastically evolving payoffs: missing proofs}
\label{app:evolving}

We prove Corollary~\ref{cor:driftingMAB-marginals} and Corollary~\ref{cor:DynamicMAB} which address the performance of contextual zooming for the stochastically evolving payoffs. In each corollary we bound from above the average dynamic regret $\hat{R}(T)$ of contextual zooming with period $T_0$, for any $T\geq T_0$. Since
    $\hat{R}(T) \leq 2 \hat{R}(T_0)$,
it suffices to bound $\hat{R}(T_0)$, which is the same as $\hat{R}(T_0)$ for (non-periodic) contextual zooming. Therefore, we can focus on analyzing the non-periodic algorithm.

We start with two simple auxiliary claims.

\begin{claim}\label{cl:Delta-Lip}
Consider the \problem{} with a product similarity space. Let
    $\Delta(x,y)\triangleq \mu^*(x) - \mu(x,y)$
be the ``badness" of point $(x,y)$ in the similarity space. Then
\begin{align}\label{eq:Delta-Lip}
|\Delta(x,y) - \Delta(x',y)| \leq 2\,\D_\TX(x,x')
    \qquad \forall x,x'\in X,\, y\in Y.
\end{align}

\end{claim}
\begin{proof}
First we show that the benchmark payoff $\mu(\cdot)$ satisfies a Lipschitz condition:
\begin{align}\label{eq:context-Lip}
|\mu^*(x) - \mu^*(x')| \leq \D_\TX(x,x')
    \qquad \forall x,x'\in X.
\end{align}
Indeed, it holds that
    $\mu^*(x) = \mu(x,y)$ and $\mu^*(x') = \mu(x,y')$
for some arms $y,y'\in Y$. Then
\begin{align*}
\mu^*(x)
    &=   \mu(x,y)
    \geq \mu(x,y')
    \geq \mu(x',y') - \D_\TX(x,x')
    = \mu^*(x') - \D_\TX(x,x'),
\end{align*}
and likewise for the other direction. Now,
\begin{align*}
|\Delta(x,y) - \Delta(x',y)|
    \leq |\mu^*(x) - \mu^*(x')| + |\mu(x,y) - \mu(x',y)|
    \leq 2\,\D_\TX(x,x').
\end{align*}
\end{proof}

\begin{claim}\label{cl:max-of-k}
Let $Z_1, \ldots, Z_k$ be independent random variables distributed uniformly at random on $[0,1]$. Let $Z^* = \max_i Z_i$.  Fix $r>0$ and let
    $S = \{i:\, Z^*> Z_i \geq Z^*-r  \}$.
Then
    $\E[\,|S|\,] = kr$.
\end{claim}

\noindent This is a textbook result; we provide a proof for the sake of completeness.

\begin{proof}
Conditional on $Z^*$, it holds that
\begin{align*}
\E[\,|S|\,]
    &= \E\left[ \textstyle{\sum_i}  \indicator{Z_i\in S} \right]
     =k\, \Pr[Z_i\in S ] \\
    &=k\, \Pr[ Z_i\in S \,| Z_i< Z^* ] \times \Pr[Z_i<Z^*] \\
    &=k\, \tfrac{r}{Z^*} \, \tfrac{k-1}{k}
     = (k-1) r/Z^*.
\end{align*}
Integrating over $Z^*$, and letting
    $F(z) \triangleq \Pr[Z^*\leq z] = z^k$,
we obtain that
\begin{align*}
\E[\,\tfrac{1}{Z^*}\,]
    &= \textstyle{\int_0^1} \tfrac{1}{z}\, F'(z) dz  = \tfrac{k}{k-1}\\
\E[\,|S|\,]
    &= (k-1) r\; \E[\,\tfrac{1}{Z^*}\,]
    = kr. \qedhere
\end{align*}
\end{proof}

\begin{proofof}{Corollary~\ref{cor:driftingMAB-marginals}}
It suffices to bound $\hat{R}(T_0)$ for (non-periodic) contextual zooming.

Let
    $\D_\TX(t,t')\triangleq \sigma |t-t'| $
be the context distance implicit in the temporal constraint~\refeq{eq:driftingMAB-Lip}. For each $r>0$, pick a number $T_r$ such that
    $\D_\TX(t,t')\leq r \iff |t-t'| \leq T_r$.
Clearly,
    $T_r \triangleq \tfrac{r}{\sigma}$.

The crux is to bound the adjusted $r$-zooming number, call it $N_r$, namely to show~\refeq{eq:driftingMAB-uniform}. For the sake of convenience, let us restate it here (and let us use the notation $T_r$):
\begin{align}\label{eq:driftingMAB-uniform-app}
\E_{\Gamma}[N_r] = O(kr)\cel{\tfrac{T}{T_r}}
    \text{ and }
\left( T_r < 1/r^2 \Rightarrow N_r=0 \right).
\end{align}

Recall that $N_r = \Nadj(\mPzoom\setminus \mWzoom)$, where $\mWzoom$ is the set of all $r$-winners (see Section~\ref{subsec:stochastic-improved} for the definition). Fix $r\in (0,1]$ and let $S$ be some $r$-packing of $\mPzoom\setminus \mWzoom$. Partition the time into $\cel{\tfrac{T}{T_r}}$ intervals of duration $T_r$.
Fix one such interval $I$. Let
    $S_I \triangleq \{(t,y)\in S:\, t\in I \}$,
the set of points in $S$ that correspond to times in $I$. Recall the notation
        $\Delta(x,y)\triangleq \mu^*(x) - \mu(x,y)$
and let
\begin{align}\label{eq:cor-driftingMAB-marginals-YI-defn}
    Y_I \triangleq \{y\in Y:\; \Delta(t_I,y) \leq  14\,r\},
        \text{ where $t_I \triangleq \min(I)$}.
\end{align}
All quantities in~\refeq{eq:cor-driftingMAB-marginals-YI-defn} refer to a fixed time $t_I$, which will allow us to use the uniform marginals property.

Note that $Y_I$ contains at least one arm, namely the best arm $y^*(t_I)$. We claim that
\begin{align}\label{eq:cor-driftingMAB-marginals-SI}
    |S_I| \leq 2\, |Y_I\setminus \{y^*(t_I)\} |.
\end{align}
Fix arm $y$. First,
    $\D_\TX(t,t')\leq r$ for any $t,t'\in I$, so
there exists at most one $t\in I$ such that $(t,y)\in S$. Second, suppose such $t$ exists. Since $S\subset \mPzoom$, it follows that
    $\Delta(t,y) \leq 12\, r$.
By Claim~\ref{cl:Delta-Lip} it holds that
\begin{align*}
\Delta(t_I,y)
    \leq \Delta(t,y) + 2\,\D_\TX(t,t')
    \leq 14\, r.
\end{align*}
So $y\in Y_I$. It follows that
    $|S_I| \leq |Y_I|$.

To obtain~\refeq{eq:cor-driftingMAB-marginals-SI}, we show that $S_I=0$ whenever $|Y_I| = 1$. Indeed, suppose $Y_I = \{y\}$ is a singleton set, and $|S_I|>0$. Then $S_I = \{(t,y)\}$ for some $t\in I$. We will show that $(t,y)$ is an $r$-winner, contradicting the definition of $S$. For any arm $y'\neq y$ and any time $t'$ such that $\D_\TX(t,t')\leq 2r$ it holds that
\begin{align*}
    \mu(t_I,y)
        &= \mu^*(t_I)
        >     \mu(t_I,y') + 14r \\
    \mu(t',y)
        &\geq \mu(t_I,y) - \D_\TX(t',t_I)
        \geq \mu(t_I,y) - 3r \\
        &>    \mu(t_I,y')+ 11r \\
        &\geq \mu(t',y') - \D_\TX(t',t_I) + 11r \\
        &\geq \mu(t',y') +8r.
\end{align*}
and so $\mu(t',y)=\mu^*(t')$. Thus, $(t,y)$ is an $r$-winner as claimed. This completes the proof of~\refeq{eq:cor-driftingMAB-marginals-SI}.

Now using~\refeq{eq:cor-driftingMAB-marginals-SI} and Claim~\ref{cl:max-of-k} we obtain that
\begin{align*}
\E_\Gamma[\, |S_I|\,]
    &\leq 2\, \E_\Gamma[\, |Y_I\setminus \{y^*(t_I)\} | \,]
    \leq O(kr) \\
\E_\Gamma[\,|S| \,]
    &\leq \cel{\tfrac{T}{T_r}} \; \E[\, |S_I|\,]
    \leq  O(kr)\; \cel{\tfrac{T}{T_r}}.
\end{align*}
Taking the $\max$ over all possible $S$, we obtain
    $\E_\Gamma[\mPzoom\setminus \mWzoom] \leq  O(kr)\; \cel{\tfrac{T}{T_r}}$.
To complete the proof of~\refeq{eq:driftingMAB-uniform-app}, we note that $S$ cannot be $r$-consistent unless $|I|\geq 1/r^2$.

Now that we have~\refeq{eq:driftingMAB-uniform-app}, the rest is a simple computation. We use Theorem~\ref{thm:ZoomAlg-adjusted}, namely we take~\refeq{eq:regret-covNum} with $r_0\rightarrow 0$, plug in~\refeq{eq:driftingMAB-uniform-app}, and recall that
    $T_r \geq 1/r^2 \iff r\geq \sigma^{1/3}$.
\begin{align*}
R(T)
    &\leq \textstyle{\sum_{r=2^i\geq \sigma^{1/3}}}\;
            \tfrac{1}{r}\ N_r\; O(\log T) \\
\E_\Gamma[R(T)]
    &\leq \textstyle{\sum_{r=2^i\geq \sigma^{1/3}}} \;
            O(k\log T)(\tfrac{T\sigma}{r}+1)  \\
    &\leq O(k\log T)(T\sigma^{2/3} + \log\tfrac{1}{\sigma}).
\end{align*}
It follows that
    $\E_\Gamma[\hat{R}(T)] \leq O(k\,\sigma^{2/3} \log T)$
for any $T\geq \sigma^{-2/3}\log \tfrac{1}{\sigma}$.
\end{proofof}

\begin{proofof}{Corollary~\ref{cor:DynamicMAB}}
It suffices to bound $\hat{R}(T_0)$ for (non-periodic) contextual zooming.

Recall that expected payoffs satisfy the temporal constraint~\refeq{eq:DynamicMAB-Lip}. Consider the high-probability event that
\begin{align}\label{eq:DynamicMAB-Lip-HP}
|\mu(t,y) - \mu(t',y)|
    \leq  \sigma\,|t-t'|^{1/2} \log T_0
    \qquad \forall t,t' \in [1, T_0],\, y\in Y.
\end{align}
Since expected regret due to the failure of~\refeq{eq:DynamicMAB-Lip-HP} is negligible, from here on we will assume that~\refeq{eq:DynamicMAB-Lip-HP} holds deterministically.

Let
    $\D_\TX(t,t') \triangleq \sigma\,|t-t'|^{1/2} \log T_0$
be the distance on contexts implicit in~\refeq{eq:DynamicMAB-Lip-HP}. For each $r>0$, define
    $T_r\triangleq (\tfrac{r}{\sigma\log T_0})^2$.
Then~\refeq{eq:driftingMAB-uniform-app} follows exactly as in the proof of Corollary~\ref{cor:driftingMAB-marginals}. We use
Theorem~\ref{thm:ZoomAlg-adjusted} similarly: we take~\refeq{eq:regret-covNum} with $r_0\rightarrow 0$, plug in~\refeq{eq:driftingMAB-uniform-app}, and note that
    $T_r \geq 1/r^2 \iff r\geq (\sigma\log T_0)^{1/2}$.
We obtain
\begin{align*}
\E_\Gamma[R(T_0)]
    &\leq \sum_{r=2^i\geq (\sigma\log T_0)^{1/2}}\;
            O(k\log T_0)(\tfrac{T_0}{T_r}+1)  \\
    &\leq O(k\log^2 T_0)(T_0\,\sigma + \log\tfrac{1}{\sigma}).
\end{align*}
It follows that
    $\E_\Gamma[\hat{R}(T)] \leq O(k\,\sigma \log^2 T_0)$
as long as $T_0\geq \tfrac{1}{\sigma}\,\log \tfrac{1}{\sigma}$.
\end{proofof}

%%%%%%%%%%%%%%%%%%%%%%%%%
%%%%%%%%%%%%%%%%%%%%%%%%%
%%%%%%%%%%%%%%%%%%%%%%%%%
\section{Contextual bandits with adversarial payoffs}
\label{sec:CovAlg}

In this section we consider the adversarial setting. We provide an algorithm which maintains an adaptive partition of the context space and thus takes advantage of ``benign" context arrivals.
It is in fact a \emph{meta-algorithm}: given a bandit algorithm \bandit, we present a contextual bandit algorithm, called \contextBandit, which calls \bandit{} as a subroutine.

%%%%%%%%%%%%%%%%%%%%%%%%
\subsection{Our setting}

Recall that in each round $t$, the context $x_t\in X$ is revealed, then the algorithm picks an arm $y_t\in Y$ and observes the payoff $\pi_t\in [0,1]$. Here $X$ is the context set, and $Y$ is the arms set. In this section, all context-arms pairs are feasible: $\mP = X\times Y$.

Adversarial payoffs are defined as follows. For each round $t$, there is a payoff function $\hat{\pi}_t: X\times Y \to [0,1]$ such that
    $\pi_t = \hat{\pi}_t(x_t,y_t)$.
The payoff function $\hat{\pi}_t$ is sampled independently from a time-specific distribution $\Pi_t$ over payoff functions. Distributions $\Pi_t$ are fixed by the adversary in advance, before the first round, and not revealed to the algorithm. Denote
    $\mu_t(x,y) \triangleq \E[\Pi_t(x,y)]$.

%The general goal is to maximize the total payoff $\sum_{t=1}^T\, \pi_t$,
% where the time horizon $T$ is fixed and known. Specifically,

Following~\cite{Hazan-colt07}, we generalize the notion of regret for context-free adversarial MAB to contextual MAB. The context-specific best arm is
\begin{align}\label{eq:bechmark-defn}
y^*(x) \in \textstyle{\argmax_{y\in Y}}\;
    \textstyle{\sum_{t=1}^T}\; \mu_t(x, y),
\end{align}
where the ties are broken in an arbitrary but fixed way. We define \emph{adversarial contextual regret} as
\begin{align}
\label{eq:regret-defn-adv}
R(T) \triangleq \textstyle{\sum_{t=1}^T}\;
	 \mu_t(x_t, y_t) - \mu^*_t(x_t) ,
	\quad\text{where}\quad
\mu^*_t(x) \triangleq \mu_t(x, y^*(x)).
\end{align}

Similarity information is given to an algorithm as a pair of metric spaces: a metric space $(X,\D_\TX)$ on contexts (the \emph{context space}) and a metric space $(Y,\D_\TY)$ on arms (the \emph{arms space}), which form the product similarity space
    $(X\times Y, \D_\TX+\D_\TY )$.
We assume that for each round $t$ functions $\mu_t$ and $\mu^*_t$ are Lipschitz on
    $(X\times Y, \D_\TX+\D_\TY )$
and $(X,\D_\TX)$, respectively, both with Lipschitz constant $1$ (see Footnote~\ref{fn:LipConst}). We assume that the context space is compact, in order to ensure that the $\max$ in~\refeq{eq:bechmark-defn} is attained by some $y\in Y$. Without loss of generality, $\mathtt{diameter}(X,\D_\TX)\leq 1$.

Formally, a problem instance consists of metric spaces $(X,\D_\TX)$ and $(Y,\D_\TY)$, the sequence of context arrivals (denoted $\arr$), and a sequence of distributions $(\Pi_t)_{t\leq T}$. Note that for a fixed distribution $\Pi_t=\Pi$, this setting reduces to the stochastic setting, as defined in Introduction. For the fixed context case ($x_t=x$ for all $t$) this setting reduces to the (context-free) MAB problem with a randomized oblivious adversary.

%%%%%%%%%%%%%%%%%%%%%%%%%%%%%
\subsection{Our results}

Our algorithm is parameterized by a regret guarantee for \bandit{} for the fixed context case, namely an upper bound on the convergence time.\footnote{The $r$-convergence time $T_0(r)$ is the smallest $T_0$ such that regret is $R(T)\leq rT$ for each $T\geq T_0$.} For a more concrete theorem statement we will assume that the convergence time of \bandit{} is at most
	$T_0(r) \triangleq c_\TY\, r^{-(2+d_\TY)}\, \log(\tfrac{1}{r})$
%	$\hat{R}_0(T) = c_\TY\, T^{-1/(2+d_\TY)} (\log T)$
for some constants $c_\TY$ and $d_\TY$ that are known to the algorithm. In particular, an algorithm in~\cite{Bobby-nips04} achieves this guarantee if $d_\TY$ is the $c$-covering dimension of the arms space and $c_\TY = O(c^{2+d_\TY})$.

This is a flexible formulation that can leverage prior work on adversarial bandits. For instance, if $Y\subset \R^d$ and for each fixed context $x\in X$ distributions $\Pi_t$ randomize over linear functions
    $\hat{\pi}_t(x,\cdot): Y\to \R$,
then one could take \bandit{} from the line of work on adversarial bandits with linear payoffs. In particular, there exist algorithms with $d_\TY=0$ and $c_\TY = \text{poly}(d)$~\citep{DaniHK-nips07,AbernethyHR-colt08,bubeck-colt12}. Likewise, for convex payoffs there exist algorithms with $d_\TY=2$ and $c_\TY = O(d)$~\citep{FlaxmanKM-soda05}. For a bounded number of arms, algorithm \EXP~\citep{bandits-exp3} achieves $d_\TY=0$ and
    $c_Y = O(\sqrt{|Y|})$.

From here on, the context space $(X,\D_\TX)$ will be only metric space considered; balls and other notions will refer to the context space only.

To quantify the ``goodness" of context arrivals, our guarantees are in terms of the covering dimension of $\arr$ rather than that of the entire context space. (This is the improvement over the guarantee~\refeq{eq:regret-naive} for the \naiveAlg.) In fact, use a more refined notion which allows to disregard a limited number of ``outliers" in $\arr$.

\begin{definition}
Given a metric space and a multi-set $S$, the \emph{$(r,k)$-covering number} of $S$ is the $r$-covering number of the set
	$\{x\in S:\, |B(x,r)\cap S|\geq k\}$.\footnote{By abuse of notation, here $|B(x,r)\cap S|$ denotes the number of points $x\in S$, with multiplicities, that lie in $B(x,r)$.}
Given a constant $c$ and a function $k:(0,1)\to \N$, the {\bf\em relaxed covering dimension} of $S$ with slack $k(\cdot)$ is the smallest $d>0$ such that the $(r, k(r))$-covering number of $S$ is at most $c\,r^{-d}$ for all $r>0$.
\end{definition}

Our result is stated as follows:

\begin{theorem}\label{thm:meta-alg}
Consider the \problem\ with adversarial payoffs, and let \bandit{} be a bandit algorithm. Assume that the problem instance belongs to some class of problem instances such that for the fixed-context case, convergence time of \bandit{} is at most
	$T_0(r) \triangleq c_\TY\, r^{-(2+d_\TY)}\, \log(\tfrac{1}{r})$
for some constants $c_\TY$ and $d_\TY$ that are known to the algorithm. Then \contextBandit{} achieves adversarial contextual regret $R(\cdot)$ such that for any time $T$ and any constant $c_\TX>0$ it holds that
\begin{align}\label{eq:meta-alg}
	R(T) \leq O(\Cdbl^2\, (c_\TX\,c_\TY)^{1/(2+d_\TX+d_\TY)})\;\,
		T^{1-1/(2+d_\TX+d_\TY)} (\log T),
\end{align} 	
where $d_\TX$ is the relaxed covering dimension of $\arr$ with multiplier $c_\TX$ and slack $T_0(\cdot)$, and $\Cdbl$ is the doubling constant of $\arr$.
\end{theorem}

\begin{note}{Remarks.}
For a version of~\refeq{eq:meta-alg} that is stated in terms of the ``raw" $(r,k_r)$-covering numbers of $\arr$, see~\refeq{eq:regret-metaAlg} in the analysis (page~\pageref{eq:regret-metaAlg}).
\end{note}

%%%%%%%%%%%%%%%%%%%%%%%%%%
\subsection{Our algorithm}

\newcommand{\ALG}{\mathtt{ALG}}
The contextual bandit algorithm \contextBandit{} is parameterized by a (context-free) bandit algorithm \bandit{}, which it uses as a subroutine, and a function
	$T_0(\cdot): (0,1)\to\N$.

The algorithm maintains a finite collection $\mA$ of balls, called \emph{active balls}. Initially there is one active ball of radius $1$. Ball $B$ stays active once it is \emph{activated}. Then a fresh instance $\ALG_B$ of \bandit{} is created, whose set of ``arms" is $Y$. $\ALG_B$ can be parameterized by the time horizon $T_0(r)$, where $r$ is the radius of $B$.

The algorithm proceeds as follows. In each round $t$ the algorithm selects one active ball $B\in \mA$ such that $x_t\in B$, calls $\ALG_B$ to select an arm $y\in Y$ to be played, and reports the payoff $\pi_t$ back to $\ALG_B$. A given ball can be selected at most $T_0(r)$ times, after which it is called \emph{full}. $B$ is called \emph{relevant} in round $t$ if it contains $x_t$ and is not full. The algorithm selects a relevant ball (breaking ties arbitrarily) if such ball exists. Otherwise, a new ball $B'$ is activated and selected. Specifically, let $B$ be the smallest-radius active ball containing $x_t$. Then
    $B'=B(x_t,\tfrac{r}{2})$,
where $r$ is the radius of $B$.  $B$ is then called the \emph{parent} of $B'$. See Algorithm~\ref{alg:contextBandit} for the pseudocode.

\begin{algorithm}[h]
\begin{algorithmic}[1]
\caption{Algorithm \contextBandit.}
\label{alg:contextBandit}
\STATE {\bf Input:}
\STATE\TAB Context space $(X,\D_\TX)$ of diameter $\leq 1$, set $Y$ of arms.
\STATE\TAB Bandit algorithm \bandit{} and a function $T_0(\cdot): (0,1)\to\N$.
\STATE {\bf Data structures:}
\STATE\TAB A collection $\mA$ of ``active balls" in $(X,\D_\TX)$.
\STATE\TAB $\forall B\in \mA$: counter $n_B$, instance $\ALG_B$ of \bandit{} on arms $Y$.
\STATE {\bf Initialization:}
\STATE\TAB
    $B\leftarrow B(x, 1)$;~~
    $\mA \leftarrow \{B\}$;~~
    $n_B\leftarrow 0$;~~
    initiate $\ALG_B$.
        \COMMENT{center $x\in X$ is arbitrary}
\STATE\TAB $\mA^* \leftarrow \mA$ \COMMENT{active balls that are not full}
\STATE {\bf Main loop:} for each round $t$\vspace{-0.3mm}
\STATE \TAB Input context $x_t$.
\STATE \TAB $\RELEVANT \leftarrow \{B\in\mA^*:\, x_t\in B\}$.
\STATE \TAB {\bf if}~~$\RELEVANT \neq \emptyset$~~{\bf then}
\STATE \TAB\TAB $B\leftarrow \text{ any } B\in \RELEVANT$.
\STATE \TAB {\bf else} \COMMENT{activate a new ball}:
\STATE \TAB\TAB $r\leftarrow \min_{B\in \mA:\; x_t\in B}\; r_B$.
\STATE \TAB\TAB $B\leftarrow B(x_t,\, r/2)$.
        \COMMENT{new ball to be added}
\STATE \TAB\TAB
    $\mA \leftarrow \mA \cup \{B\}$;~~
    $\mA^* \leftarrow \mA^* \cup \{B\}$;~~
    $n_B\leftarrow 0$;~~
    initiate $\ALG_B$.
\STATE \TAB $y \leftarrow \text{next arm selected by $\ALG_B$}$.
\STATE \TAB Play arm $y$, observe payoff $\pi$, report $\pi$ to $\ALG_B$.
\STATE \TAB $n_B\leftarrow n_B+1$.
\STATE \TAB {\bf if} $n_B = T_0(\mathtt{radius}(B))$ {\bf then}
    $\mA^* \leftarrow \mA^* \setminus \{B\}$.
    \COMMENT{ball $B$ is full}
\end{algorithmic}
\end{algorithm}

\subsection{Analysis: proof of Theorem~\ref{thm:meta-alg}}

First let us argue that algorithm \contextBandit{} is well-defined. Specifically, we need to show that after the activation rule is called, there exists an active non-full ball containing $x_t$. Suppose not. Then the ball
    $B'=B(x_t,\tfrac{r}{2})$
activated by the activation rule must be full. In particular, $B'$ must have been active before the activation rule was called, which contradicts the minimality in the choice of $r$.  Claim proved.

We continue by listing several basic claims about the algorithm.

\begin{claim} \label{cl:metaAlg-basic}
The algorithm satisfies the following basic properties:
\begin{OneLiners}
\item[(a)] (Correctness) In each round $t$, exactly one active ball is selected.
\item[(b)] Each active ball of radius $r$ is selected at most $T_0(r)$ times.
\item[(c)] (Separation) For any two active balls $B(x,r)$ and $B(x',r)$ we have $\D_\TX(x,x')> r$.
\item[(d)] Each active ball has at most $\Cdbl^2$ children, where
$\Cdbl$ is the doubling constant of $\arr$.
\end{OneLiners}
\end{claim}

\begin{proof}
Part (a) is immediate from the algorithm's specification. For (b), simply note that by the algorithms' specification a ball is selected only when it is not full.

To prove (c), suppose that $\D_\TX(x,x')\leq r$ and suppose $B(x',r)$ is activated in some round $t$ while $B(x,r)$ is active. Then $B(x',r)$ was activated as a child of some ball $B^*$ of radius $2r$. On the other hand, $x'=x_t \in B(x,r)$, so $B(x,r)$ must have been full in round $t$ (else no ball would have been activated), and consequently the radius of $B^*$ is at most $r$. Contradiction.

For (d), consider the children of a given active ball $B(x,r)$. Note that by the activation rule the centers of these children are points in $\arr \cap B(x,r)$, and by the separation property any two of these points lie within distance $>\tfrac{r}{2}$ from one another. By the doubling property, there can be at most $\Cdbl^2$ such points.
\end{proof}

Let us fix the time horizon $T$, and let $R(T)$ denote the contextual regret of \contextBandit. Partition $R(T)$ into the contributions of active balls as follows. Let $\mB$ be the set of all balls that are active after round $T$. For each $B\in\mB$, let $S_B$ be the set of all rounds $t$ when $B$ has been selected. Then
\begin{align*}
R(T)    = \textstyle{\sum_{B\in \mB}}\; R_B(T),
	\quad\text{where}\quad
R_B(T)  \triangleq \textstyle{\sum_{t\in S_B}}\;
    \mu^*_t (x_t) - \mu_t (x_t,\, y_t).
\end{align*}

\begin{claim}\label{cl:regret-ball}
For each ball $B = B(x,r) \in \mB$,  we have
$R_B \leq 3\,r\, T_0(r)$.
\end{claim}

\begin{proof}
By the Lipschitz conditions on $\mu_t$ and $\mu^*_t$, for each round $t\in S_B$ it is the case that
\begin{align*}
\mu^*_t(x_t)
     \leq r + \mu^*_t(x)
     = r + \mu_t(x,y^*(x))
     \leq 2rn + \mu_t(x_t,\, y^*(x)).
\end{align*}
The $t$-round regret of \bandit{} is at most
	$R_0(t) \triangleq t\,T_0^{-1}(t)$.
Therefore, letting $n = |S_B|$ be the number of times algorithm $\ALG_B$ has been invoked, we have that
\begin{align*}
R_0(n) + \textstyle{\sum_{t\in S_B}}\;  \mu_t(x_t,\, y_t)
	\geq \textstyle{\sum_{t\in S_B}}\; \mu_t(x_t, y^*(x))
	\geq \textstyle{\sum_{t\in S_B}}\; \mu^*_t(x_t) - 2rn.
\end{align*}
Therefore
	$R_B(T)\leq R_0(n) + 2rn$.
Recall that by Claim~\ref{cl:metaAlg-basic}(b) we have $n\leq T_0(r)$. Thus, by definition of convergence time
	$R_0(n)\leq R_0(T_0(r)) \leq r\, T_0(r)$,
and therefore $R_B(T)\leq 3\,r\,T_0(r)$.
\end{proof}

Let $\F_r$ be the collection of all full balls of radius $r$. Let us bound $|\F_r|$ in terms the $(r,k)$-covering number of $\arr$ in the context space, which we denote $N(r,k)$.

% $k=T_0(r)/\Cdbl^2$

\begin{claim}\label{cl:full-context-balls}
There are at most $N(r,\, T_0(r))$ full balls of radius $r$.
\end{claim}

\begin{proof}
Fix $r$ and let $k = T_0(r)$. Let us say that a point $x\in \arr$ is \emph{heavy} if $B(x,r)$ contains at least $k$  points of $\arr$, counting multiplicities. Clearly, $B(x,r)$ is full only if its center is heavy. By definition of the $(r,k)$-covering number, there exists a family $\mathcal{S}$ of $N(r, k)$ sets of diameter $\leq r$ that cover all heavy points in $\arr$. For each full ball $B = B(x,r)$, let $S_B$ be some set in $\mathcal{S}$ that contains $x$. By Claim~\ref{cl:metaAlg-basic}(c), the sets $S_B$, $B\in \F_r$ are all distinct. Thus,
	$|\F_r| \leq |\mathcal{S}| \leq  N(r,k)$.
\end{proof}

Let $\mB_r$ be the set of all balls of radius $r$ that are active after round $T$.
By the algorithm's specification, each ball in $\F_r$ has been selected $T_0(r)$ times, so
	$|\F_r| \leq T/T_0(r)$.
Then using Claim~\ref{cl:metaAlg-basic}(b) and Claim~\ref{cl:full-context-balls}, we have
\begin{align}
|\mB_{r/2}|
	&\leq \Cdbl^2\, |\F_r|
	\leq \Cdbl^2\; \min(T/T_0(r),\; N(r,\, T_0(r))) \nonumber \\
\textstyle{\sum_{B\in \mB_{r/2}}} R_B
	&\leq O(r) \, T_0(r)\,|\mB_{r/2}|
	\leq O(\Cdbl^2)\; \min(rT,\; r\, T_0(r)\, N(r, T_0(r))). \label{eq:regret-metaAlg-r}
\end{align}

\noindent Trivially, for any full ball of radius $r$ we have $T_0(r)\leq T$. Thus, summing~\refeq{eq:regret-metaAlg-r} over all such $r$, we obtain
\begin{align}\label{eq:regret-metaAlg}
R(T) \leq O(\Cdbl^2)\;
		\textstyle{\sum_{r=2^{-i}:\, i\in \N \text{ and } T_0(r)\leq T }}  \;
			 \min(rT,\; r\,T_0(r)\, N(r, T_0(r))).
\end{align}
Note that~\refeq{eq:regret-metaAlg} makes no assumptions on $N(r, T_0(r))$. Now, plugging in
	$T_0(r) = c_\TY\, r^{-(2+d_\TY)}$ and $N(r,T_0(r)) \leq c_\TX\, r^{-d_\TX}$ into~\refeq{eq:regret-metaAlg} and optimizing it for $r$ it is easy to derive the desired bound~\refeq{eq:meta-alg}.

\section{Conclusions}

We consider a general setting for contextual bandit problems where the algorithm is given information on similarity between the context-arm pairs. The similarity information is modeled as a metric space with respect to which expected payoffs are Lipschitz-continuous. Our key contribution is an algorithm which maintains a partition of the metric space and adaptively refines this partition over time. Due to this ``adaptive partition" technique, one can take advantage of ``benign" problem instances without sacrificing the worst-case performance; here ``benign-ness" refers to both expected payoffs and context arrivals. We essentially resolve the setting where expected payoff from every given context-arm pair either does not change over time, or changes slowly. In particular, we obtain nearly matching lower bounds (for time-invariant expected payoffs and for an important special case of slow change).

We also consider the setting of adversarial payoffs. For this setting, we design a different algorithm that maintains a partition of contexts and adaptively refines it so as to take advantage of ``benign" context arrivals (but not ``benign" expected payoffs), without sacrificing the worst-case performance. Our algorithm can work with, essentially, any given off-the-shelf algorithm for standard (non-contextual) bandits, the choice of which can then be tailored to the setting at hand.

The main open questions concern relaxing the requirements on the quality of similarity information that are needed for the provable guarantees. First, it would be desirable to obtain similar results under weaker versions of the Lipschitz condition. Prior work~\citep{LipschitzMAB-stoc08,xbandits-nips08}
obtained several such results for the non-contextual version of the problem, mainly because their main results do not require the full power of the Lipschitz condition. However, the analysis in this paper appears to make a heavier use of the Lipschitz condition; it is not clear whether a meaningful relaxation would suffice. Second, in some settings the available similarity information might not include any numeric upper bounds on the difference in expected payoffs; e.g. it could be given as a tree-based taxonomy on context-arm pairs, without any explicit numbers. Yet, one wants to recover the same provable guarantees \emph{as if} the numerical information were explicitly given. For the non-contextual version, this direction has been explored in ~\citep{Bubeck-alt11,ImplicitMAB-nips11}.%
\footnote{\citep{Bubeck-alt11,ImplicitMAB-nips11} have been published after the preliminary publication of this paper on {\tt arxiv.org}.}

Another open question concerns our results for adversarial payoffs. Here it is desirable to extend our ``adaptive partitions" technique to also take advantage of ``benign" expected payoffs (in addition to ``benign" context arrivals). However, to the best of our knowledge such results are not even known for the non-contextual version of the problem.

\xhdr{Acknowledgements.} The author is grateful to Ittai Abraham, Bobby Kleinberg and Eli Upfal for many conversations about multi-armed bandits, and to Sebastien Bubeck for help with the manuscript. Also, comments from anonymous COLT reviewers and JMLR referees have been tremendously useful in improving the presentation.

%\baselineskip
% ############ BIBLIOGRAPHY ##############
\begin{small}
\bibliography{bib-abbrv,bib-bandits,bib-embedding,bib-slivkins}
\end{small}
\appendix

\section{The KL-divergence technique, encapsulated}
\label{app:KL-div}

To analyze the lower-bounding construction in Section~\ref{sec:LBs}, we use an extension of the KL-divergence technique from~\cite{bandits-exp3},  which is implicit in~\cite{Bobby-nips04} and encapsulated as a stand-alone theorem in \cite{LipschitzMAB-merged-arxiv}. To make the paper self-contained, we state the theorem from \cite{LipschitzMAB-merged-arxiv}, along with the relevant definitions. The remainder of this section is copied from \cite{LipschitzMAB-merged-arxiv}, with minor modifications.

\vspace{2mm}

Consider a very general MAB setting where the algorithm is given a strategy set $X$ and a collection $\F$ of feasible payoff functions; we call it the \emph{feasible MAB problem} on $(X,\F)$. For example, $\F$ can consist of all functions $\mu:X\to [0,1]$ that are Lipschitz with respect to a given metric space. The lower bound relies on the existence of a collection of subsets of $\F$ with certain properties, as defined below. These subsets correspond to children of a given tree node in the ball-tree

\begin{definition}\label{def:eps-k-ensemble}
Let $X$ be the strategy set and $\F$ be the set of all feasible payoff functions. An \emph{$(\eps,k)$-ensemble} is
 a collection of subsets $\F_1 \LDOTS \F_k \subset \F$ such that there exist mutually disjoint subsets
        $S_1 \LDOTS S_k \subset X$ and a number $\mu_0 \in [\tfrac13, \tfrac23]$
which satisfy the following. Let
        $S=\cup_{i=1}^k S_i$.
Then
\begin{OneLiners}
\item on $X\setminus S$, any two functions in $\cup_i\, \F_i $ coincide, and are bounded from above by $\mu_0$.
\item for each $i$ and each function $\mu\in \F_i$ it holds that $\mu = \mu_0$ on $S\setminus S_i$ and
    $\sup(\mu_i, S_i) = \mu_0+\eps$.
\end{OneLiners}
\end{definition}

Assume the payoff function $\mu$ lies in $\cup_i\, \F_i$. The idea is that an algorithm needs to play arms in $S_i$ for at least $\Omega(\eps^{-2})$ rounds in order to determine whether $\mu\in\F_i$, and each such step incurs $\eps$ regret if $\mu\not\in\F_i$. In our application, subsets $S_1 \LDOTS S_k$ correspond to children $u_1 \LDOTS u_k$  of a given tree node in the ball-tree, and each $\F_i$ consists of payoff functions induced by the ends in the subtree rooted at $u_i$.

\begin{theorem}[Theorem 5.6 in \cite{LipschitzMAB-merged-arxiv}]
\label{thm:KL-div}
Consider the feasible MAB problem with 0-1 payoffs. Let  $\F_1, \ldots, \F_k$ be an $(\eps,k)$-ensemble,
where $k\geq 2$ and $\eps\in(0,\,\tfrac{1}{12})$. Then for any
    $t \leq \tfrac{1}{32}\, k\,\eps^{-2}$
and any bandit algorithm there exist at least $k/2$ distinct $i$'s such that the regret of this algorithm on any payoff function from $\F_i$ is at least $\tfrac{1}{60}\,\eps t$.
\end{theorem}

In ~\cite{bandits-exp3}, the authors analyzed a special case of an $(\eps,k)$-ensemble in which there are $k$ arms
	$u_1 \LDOTS u_k$,
and each $\F_i$ consists of a single payoff function that assigns expected payoff $\tfrac12+\eps$ to arm $u_i$, and $\tfrac12$ to all other arms.

\end{document}